\documentclass[a4paper,USenglish,thm-restate]{lipics-v2019}
\usepackage[utf8]{inputenc}
\usepackage{amsmath}
\usepackage{amsthm}
\usepackage{graphicx}
\usepackage{wrapfig}
\usepackage[noend]{algpseudocode}
\usepackage{makecell}
\usepackage{xspace}
\usepackage{caption}
\usepackage{imakeidx}
\usepackage[dvipsnames]{xcolor}
\graphicspath{{./figures/}}
\usepackage{thmtools}
\usepackage{thm-restate}
\theoremstyle{plain}

\usepackage[ruled,vlined,noend]{algorithm2e}
\usepackage{marvosym}

\newtheorem{observation}{Observation}
\let\emptyset\varnothing
\newcommand{\etal}{\textit{et al.}\xspace}

\newcommand{\BR}{\ensuremath{\mathbb B}_{\mathcal{{R}}}}

\newcommand{\RR}{\ensuremath{\mathcal R}\xspace}
\newcommand{\RRo}{\ensuremath{\mathcal R^{\text{\tiny{\textregistered}} }\xspace} }
\newcommand{\RRt}{\ensuremath{\mathcal R^{ \text{\LeftScissors} }  }\xspace}
\newcommand{\RRc}{\ensuremath{\mathcal R^{\star} }\xspace}
\newcommand{\Oh}{\ensuremath{O}}
\newcommand{\pareto}{Pareto front\xspace}
\newcommand{\CP}{\mathrm{CP}}

\newcommand{\Vnext}[1]{\ensuremath{V_{#1}^{\textit{next} }}\xspace} 
\newcommand{\Hprev}[1]{\ensuremath{H_{#1}^{\textit{prev} }}\xspace} 
\newcommand{\RP}{\ensuremath{ \tilde{\mathcal{\RR}}(P) }}
\newcommand{\subproblem}{subproblem\xspace}
\newcommand{\subproblems}{subproblems\xspace}

\nolinenumbers
\newtheorem{invariant}{Invariant}
\hideLIPIcs

\title{Preprocessing Imprecise Points for the Pareto~Front} 
\titlerunning{Preprocessing Imprecise Points for the Pareto Front}

\author{Ivor van der Hoog}{Department of Information and Computing Sciences, Utrecht University, the Netherlands}{i.d.vanderhoog@uu.nl}{}{Supported by the Dutch Research Council (NWO); 614.001.504.}
\author{Irina Kostitsyna}{Department of Mathematics and Computer Science, TU Eindhoven, the Netherlands}{i.kostitsyna@tue.nl}{}{}
\author{Maarten L\"{o}ffler}{Department of Information and Computing Sciences, Utrecht University, the Netherlands}{m.loffler@uu.nl}{}{Partially supported by the Dutch Research Council (NWO); 614.001.504.}
\author{Bettina Speckmann}{Department of Mathematics and Computer Science, TU Eindhoven, the Netherlands}{b.speckmann@tue.nl}{https://orcid.org/0000-0002-8514-7858}{Partially supported by the Dutch Research Council (NWO); 639.023.208.}

\authorrunning{I. van der Hoog, I. Kostitsyna, M. L\"{o}ffler, B. Speckmann.}

\Copyright{Ivor van der Hoog, Irina Kostitsyna, Maarten L\"{o}ffler, Bettina Speckmann} 


\begin{document}

\maketitle

\begin{abstract}
In the preprocessing model for uncertain data we are given a set of regions $\RR$ which model the uncertainty associated with an unknown set of points $P$. In this model there are two phases: a preprocessing phase, in which we have access only to $\RR$, followed by a reconstruction phase, in which we have access to points in $P$ at a certain retrieval cost $C$ per point.
We study the following algorithmic question: how fast can we construct the \pareto of $P$ in the preprocessing model?

We show that if $\RR$ is a set of pairwise-disjoint axis-aligned rectangles, then we can preprocess $\RR$ to reconstruct the Pareto front of $P$ efficiently. To refine our algorithmic analysis, we introduce a new notion of algorithmic optimality which relates to the entropy of the uncertainty regions. Our proposed \emph{uncertainty-region optimality} falls on the spectrum between worst-case optimality and instance optimality.  
We prove that instance optimality is unobtainable in the preprocessing model, whenever the classic algorithmic problem reduces to sorting.
Our results are worst-case optimal in the preprocessing phase; in the reconstruction phase, our results are uncertainty-region optimal with respect to real RAM instructions, and instance optimal with respect to point retrievals.
\end{abstract}

\newpage
\section{Introduction}
In many applications of geometric algorithms to real-world problems the input is inherently imprecise. 
A classic example are GPS samples used in GIS applications, which have a significant error.
Geometric imprecision can be caused by other factors as well. 
For example, if a measured object moves during measurement, it may have an error dependent on its speed~\cite {kahan1992real}. 
Another example comes from I/O-sensitive computations: exact locations may be too costly to store in local memory~\cite{bruce2005efficient}. 
Algorithms that can handle imprecise input well have received considerable attention in computational geometry. 
We continue this line of research by studying
the efficient construction of the \pareto of a collection of imprecise points.

\subparagraph{Preprocessing model.} Held and Mitchell~\cite{held2008triangulating} introduced the \emph{preprocessing model of uncertainty} as a model to study the amount of geometric information contained in uncertain points. In this model, the input is a set of geometric (uncertainty) regions $\RR = ( R_1, R_2, \ldots, R_n )$ with an associated ``true'' planar point set $P = (p_1, p_2, \ldots, p_n )$. 
For any pair $(\RR, P)$, we say that $P$ \emph{respects} $\RR$ if each $p_i$ lies inside its associated region $R_i$; we assume throughout the paper that $P$ respects $\RR$.
The preprocessing model has two consecutive phases: a preprocessing phase where we have access only to the set of uncertainty regions $\RR$ and a reconstruction phase where  we can for each $R_i \in \RR$, request the true location $p_i$ in (traditionally constant) $C$ time. The value $C$ can, for example, model the cost of disk retrievals for I/O-sensitive computations~\cite{bruce2005efficient}.
We typically want to preprocess $\RR$ in $\Oh(n \log n)$ time to create some linear-size auxiliary datastructure $\Xi$. Afterwards, we want to reconstruct the desired output on $P$ using $\Xi$ faster than would be possible without preprocessing. 

L{\"o}ffler and Snoeyink~\cite{loffler2010delaunay} were the first to interpret $\RR$ as a collection of imprecise measurements of a true point set $P$. 
The size of $\Xi$ and the running time of the reconstruction phase, together quantify the information about (the Delaunay triangulation of) $P$ contained in $\RR$. 
This interpretation was widely adopted within computational geometry and motivated many recent results for constructing Delaunay triangulations~\cite {buchin2009delaunay,buchin2011delaunay,devillers2011delaunay, van2010preprocessing},
spanning trees~\cite{liu2018minimizing,zeng2016integrating}, convex hulls~\cite{evans2011possible, ezra2013convex, loffler2010largest,nagai1999convex}  and other planar decompositions~\cite{loffler2013unions, van2019preprocessing} for imprecise points.

\begin{figure}[b]
    \centering
    \includegraphics{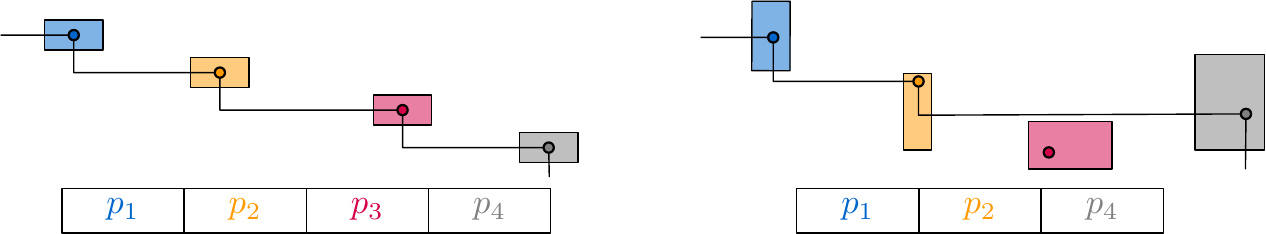}
    \caption{The \pareto of $P$ can be implied by the geometry of $\RR$ (left) or not (right). }
    \label{fig:overlap}
    \vspace{0cm}
\end{figure}

\subparagraph{Output format.}
Classical work in the preprocessing model ultimately aims to preprocess the data in such a way that one can achieve a (near-){\em linear}-time reconstruction phase.
Indeed, if the final output structure has linear complexity and must explicitly contain the coordinates of each value in $P$, then returning the result takes $\Omega(nC)$ time.
However, this point of view is limiting in two ways.
First, certain geometric problems, such as the convex hull or the \pareto, may have sub-linear output complexity.
Second, even if the output has linear complexity, it may be possible to find its combinatorial structure without inspecting the true locations of all points.
Consider the  example in Figure~\ref{fig:overlap}:
on the left, we do not need to retrieve any point; on the right, we do not need to retrieve $p_3$ after we retrieve $p_4$. 
Van der Hoog~\etal~\cite{van2019preprocessing} propose an addition to the preprocessing model to enable a more fine-grained analysis in these situations: instead of returning the desired structure on $P$ explicitly, they instead return an \emph{implicit representation} of the output. 
This implicit representation can take the form of a pointer structure which is guaranteed to be isomorphic to the desired output on $P$, but where each value is a pointer to either a certain (retrieved) point, or to an uncertain (unretrieved) point.
In this paper, we study the efficient construction of the \pareto of a set of imprecise points $P$, from pairwise-disjoint axis-aligned rectangles $\RR$ as uncertainty regions,
in the preprocessing model with implicit representation.

\subparagraph{Algorithmic efficiency.} 
To assess the efficiency of any algorithm we generally want to compare its performance to a suitable lower bound. 
Two common types of lower bounds are \emph{worst-case} and \emph{instance} lower bounds. The classical worst-case lower bound takes the minimum over all algorithms $A$, of the maximal running time of $A$ for any pair $(\RR, P)$.
The instance lower bound~\cite{afshani2017instance, evans2013competitive} is the minimum over all $A$, for a fixed instance $(\RR, P)$, of the running time of $A$ on $(\RR, P)$. For the \pareto 
the worst-case lower bound is trivially $\Omega(nC)$; worst-case optimal performance (for us, in the reconstruction phase) is hence easily obtainable. Instance-optimality, on the other hand, is unobtainable in classical computational geometry~\cite{afshani2017instance}. Consider, for example, binary search for a value $q$ amongst a set $X$ of sorted numbers. 
For each instance $(X, q)$, there exists a naive algorithm that guesses the correct answer in constant time. Thus the instance lower bound for binary search is constant, even though there is no algorithm that can perform binary search in constant time in a comparison-based RAM model~\cite{hoog2020FOCS}.
Hence we introduce a new lower bound for the preprocessing model, whose granularity falls in between the instance and worst-case lower bound. Our \emph{uncertainty-region} lower bound is the minimum over all algorithms $A$, for a fixed input $\RR$, of the maximal running time of $A$ on $(\RR, P)$ for any $P$ that respects~$\RR$. A detailed discussion of algorithmic efficiency for the preprocessing model 
can be found in Section~\ref{sec:optimality}.

\subparagraph{Related work.}
Bruce~\etal~\cite{bruce2005efficient} study the efficient construction of the \pareto of two-dimensional pairwise disjoint axis-aligned uncertainty rectangles in what would later be the preprocessing model using implicit representation.
As their paper is motivated by I/O-sensitive computation, they assume that the retrieval cost $C$ dominates polynomial RAM running time and both their preprocessing and reconstruction phase use an unspecified polynomial number of RAM instructions.
In the reconstruction phase they have a retrieval-strategy that iteratively selects a region $R_i$ for which they retrieve $p_i$ to construct $\Xi^*$ (since $\Xi^*$ is an \emph{implicit representation}, they do not have to retrieve each $p_i \in P$). Their result is instance optimal under their assumption that $C$ dominates the RAM running time of all parts of their algorithm. We study the same problem without their assumption on $C$.

\subparagraph{Results and organization} We discuss in Section~\ref{sec:optimality} the three possible lower bounds for the preprocessing model: worst case, instance, and our new uncertainty-region lower bound. In Section~\ref{sec:prelims} we present the necessary geometric preliminaries.
Then, in Section~\ref{sec:lowerbound}, we prove an uncertainty-region lower bound on the time required for the reconstruction phase. 
In Section~\ref{sec:upperbound} we then show how to preprocess $\RR$ in $\Oh(n \log n)$ time to create an auxiliary structure $\Xi$. We also explain how to reconstruct the \pareto of $P$ as an implicit representation $\Xi^*$ from $\Xi$. 
Our results are worst-case optimal in the preprocessing phase; our reconstruction results are uncertainty-region optimal in the RAM instructions, instance optimal with respect to the retrieval cost $C$ and an $O(\log n)$ factor removed from instance optimal with respect to both. This is the first two-dimensional result in the preprocessing model with better than worst-case optimal performance.

\section{Algorithmic optimality}
\label{sec:optimality}

We briefly revisit the definitions of worst-case and instance lower bounds in the preprocessing model and then formally introduce our new uncertainty-region lower bound.

\subparagraph{Worst-case lower bounds.}
The worst-case comparison-based lower bound of an algorithmic problem $\mathcal{P}$ considers each
algorithm\footnote{We refer to comparison-based algorithms algorithms on an intuitive level: as RAM computations that do not make use of flooring. For a more formal definition we refer to any of~\cite{afshani2017instance, ben1983lower, demaine2020finding, hoog2020FOCS}. }  plus datastructure pair $(A, \Xi)$ which solves $\mathcal{P}$ in a \emph{competitive setting} with respect to their 
maximal running time: 
\[
\textnormal{Worst-case lower bound}(\mathcal{P}) := \min_{(A, \Xi)} \max_{(\RR, P) } \textnormal{Runtime}(A, \Xi, \RR, P)\,.
\]
\noindent
The number $L$ of distinct outcomes for all instances $(\RR, P)$ implies a lower bound on the maximal running time for any algorithm $A$: regardless of preprocessing, auxiliary datastructures and memory used, any comparison-based pointer machine algorithm $A$  can be represented as a decision tree where at each algorithmic step, a binary decision is taken~\cite{ben1983lower, cardinal2019information, hoog2020FOCS}. Since there are at
least $L$ different outcomes, there must exists a pair $(\RR, P)$ for which $A$ takes $\log L$ steps before $A$ terminates (this lower bound is often referred to as the \emph{information theoretic lower bound} or sometimes the \emph{entropy} of the problem~\cite{afshani2017instance, cardinal2005minimum, cardinal2019information}). 

\begin{figure}[t]
    \centering
    \includegraphics{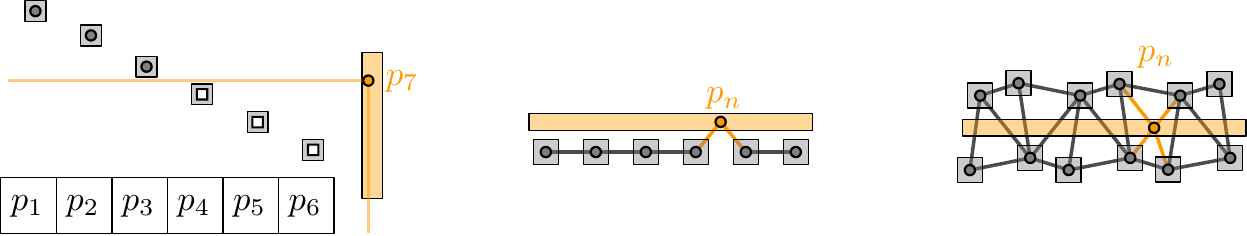}
    \caption{
    Thrice a collection of grey uncertainty regions where the Pareto front, EMST or Delaunay triangulation of the grey points is implied by the regions; plus an orange region $R_n$. Depending on the placement of $p_n$, it can neighbor any grey point in the final structure. }
    \label{fig:impossible}
    \vspace{-0.3cm}
\end{figure}

\subparagraph{Instance lower bounds.}
A stronger lower bound, is an instance lower bound~\cite{evans2013competitive} (or \emph{instance optimal in the random-order setting} in \cite{afshani2017instance}). 
For an extensive overview of instance optimality we refer to Appendix~\ref{appx:entropy}.
For a given instance $(\RR, P)$, its instance lower bound is:
\[
 \textnormal{Instance lower bound}(\mathcal{P}, \RR, P) = \min_{(A, \Xi)} \textnormal{Runtime}(A, \Xi, \RR, P)\,.
\]
An algorithm $A$ is instance optimal, if for every instance $(\RR,P)$ the runtime of $A$ matches the instance lower bound.
L\"{o}ffler \etal~\cite{loffler2013unions} define \emph{proximity structures} that include quadtrees, Delaunay triangulations, convex hulls, Pareto fronts and Euclidean minimum spanning trees. We prove the following:

\begin{restatable}{theorem}{noinstance}
Let the unspecified retrieval cost $C$ not dominate $O(\log n)$ RAM instructions and $\RR$ be any set of pairwise disjoint uncertainty rectangles.
Then there exists no algorithm $A$ in the preprocessing model with implicit representation that can construct a proximity data structure on the true points which is instance optimal.
\end{restatable}

\begin{proof}
Let $\RR' = (R_1, R_2, \ldots R_{n-1})$ be a set of uncertainty regions for which the implicit data structure $\Xi^*$ can be known in the preprocessing phase. Denote by $R_n$ an uncertainty region for which $p_n$ can neighbor any $p_i \in (p_1, \ldots p_{n-1})$. See Figure~\ref{fig:impossible} for an example of the Pareto front, the EMST and the Delaunay triangulation (with it, Voronoi diagrams) and Figure~\ref{fig:CH} for the convex hull. 
For the set of grey points $(p_1, \ldots p_{n-1})$, their respective structure is known while the orange point $p_n$ can neighbor any of the grey points. 
Via the information theoretic lower bound, there is no algorithm~$A$ that for every instance can decide the correct neighbor of $p_n$ in $O(C)$ time. Yet for every instance, there exists a naive algorithm that correctly guesses the constantly many neighbors of $p_n$ and verifies this guess in $O(C)$ time. 
\end{proof}

\begin{figure}[h]
    \centering
    \includegraphics{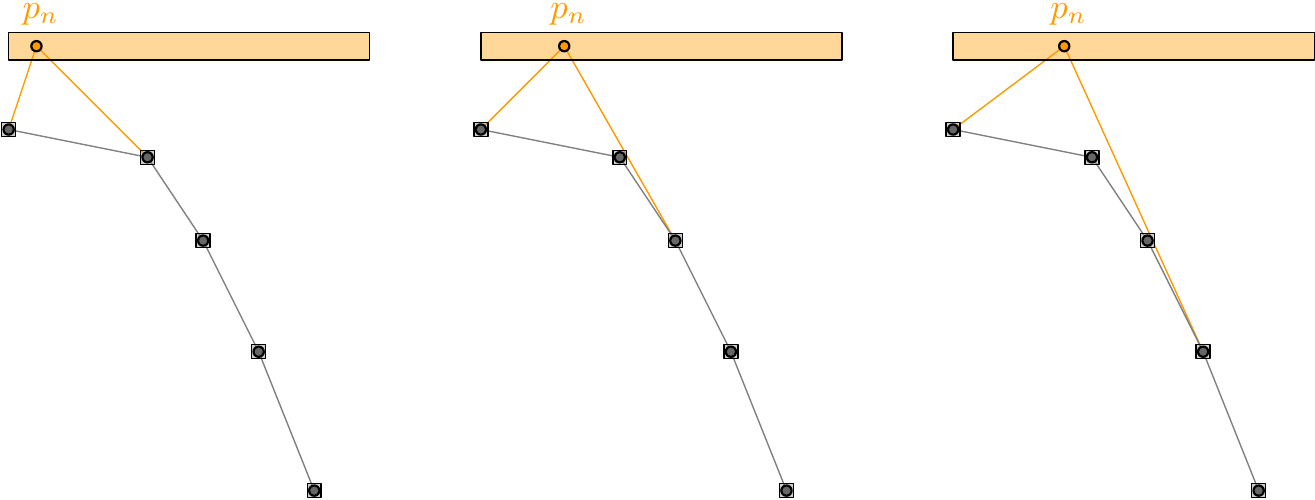}
    \caption{
    A collection of $n-1$ grey pairwise-disjoint uncertainty rectangles, for which the convex hull of their underlying points is implied by the convex hull of their bottom left vertices. The region $R_n$ is shown in orange. Depending on the placement of $p_n$, it can neighbor any grey point in the convex hull of all the points.  }
    \label{fig:CH}
\end{figure}

\subparagraph{Uncertainty-region lower bounds.}
Worst-case optimality is easily attainable by any algorithm and we proved that instance optimality is not attainable in the preprocessing model.
Yet the examples in Figure~\ref{fig:overlap}~and~\ref{fig:impossible} intuitively have a lower bound of $\Theta(1)$ and $\Theta(\log n + C)$, which is trivial to match via binary search. 
We capture this intuition for a fixed input $\RR$:
\[
\textnormal{Uncertainty-region lower bound}(\mathcal{P}, \RR) := \min_{(A, \Xi)} \max_{ (P \textnormal{ respects } \RR ) } \textnormal{Runtime}(A, \Xi, \RR, P)\,,
\]
and say an algorithm $A$ is uncertainty-region optimal if for every $\RR$, $A$ has a running time that matches the uncertainty-region lower bound. Denote by $L(\RR)$ the number of distinct outcomes for all $P$ that respect $\RR$. 
Via the information theoretic lower bound we know: 
\[
\forall \RR, \quad \log |L(\RR)| \le \textnormal{Uncertainty-region lower bound}(\mathcal{P}, \RR)\,.
\]
For constructing proximity structures in the preprocessing model with implicit representations, the value of $\log L(\RR)$ can range from anywhere between $0$ and $n \log n$.
Consequently, an optimal algorithm cannot necessarily afford to explicitly retrieve the entire point set $P$.

\section{Geometric preliminaries}
\label{sec:prelims}
Throughout the paper, we use the notation $\RRo$, $\RRt$ for original and truncated regions respectively (which we define later).
When the set is clear from context, we drop the superscript.
Let $\RR = (R_1, R_2, \ldots, R_{n})$ be a sequence of $n$ pairwise disjoint 
closed
axis-aligned uncertainty rectangles, with underlying point set $P$. For ease of exposition, we assume $\RR$ and $P$ lie in general position (no points or region vertices share a coordinate).
We denote by $[R_i, R_j] := (R_i, R_{i+1}, \ldots, R_j)$ a subsequence of $j- i + 1$ regions and similarly by $[p_i, p_j] = (p_i, p_{i+1}, \ldots, p_j)$ a subsequence of points.
For brevity, with slight abuse of notation, we may refer to points as degenerate rectangles; hence any set $\RR$ may contain points. 
Whenever we place points on a vertex, we mean placing it arbitrarily close to said vertex. A region $R_i$ precedes a region $R_j$ if $i < j$. 
Conversely, $R_j$ succeeds $R_i$.

For two points $p$ and $q$, we say that $p$ \emph{(Pareto) dominates} $q$ if both its $x$- and $y$-coordinates are greater than or equal to the respective coordinates of $q$.
A point $p$ \emph{(Pareto) dominates} a rectangle $R$, if $p$ dominates its top right vertex.
We define the \emph{\pareto} of $P$ as the boundary of the set of points that are dominated by a point in $P$.
That is, the \pareto is the set of points in $P$ that are not dominated by any other point in $P$, connected by a rectilinear staircase. 
For any region or point $R$, we define its \emph{horizontal halfslab} as the union of all horizontal halflines that are directed leftward, whose apex lies in or on $R$. We define the \emph{vertical halfslab} symmetrically using downward vertical halflines.
Given a set $\RR$ without knowledge of $P$,
we say a region $R_i \in \RR$ is (Figure~\ref{fig:truncation}, left):
\begin {itemize}
  \item a \emph{negative} region if for all choices of $P$, the point $p_i$ is not part of the \pareto of $P$;
  \item a \emph{positive} region if for all choices of $P$, the point $p_i$ is part of the \pareto; or
  \item a \emph{potential} region if it is neither positive nor negative.
\end {itemize}

\begin{figure}[tb]
    \centering
    \includegraphics{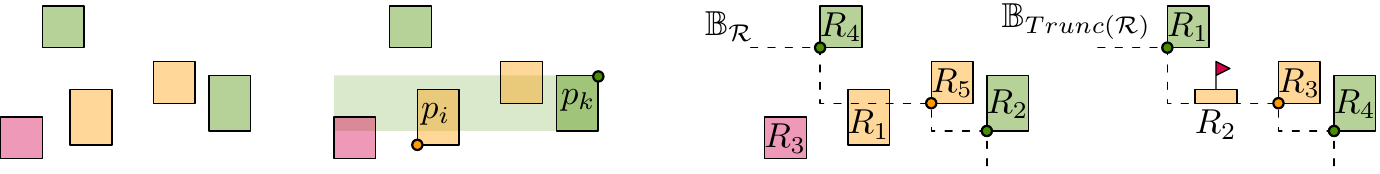}
    \caption{Left: a collection of uncertainty regions. Green is positive, red is negative and yellow is potential. The horizontal halfslab of a green region is shown. Right: A collection of uncertainty regions before and after truncation, note that we re-indexed the regions and flagged one. }
    \label{fig:truncation}
    \vspace{0cm}
\end{figure}

\begin{restatable}{lem}{classification}
\label{lemma:classification}
A region $R_i \in \RR$ is negative if and only if $\exists R_j \in \RR$ such that the top right vertex of $R_i$ is  dominated by the bottom left vertex of $R_j$. A non-negative region $R_i$ is positive if and only if $\not \exists R_k \in \RR$ such that $R_i$ intersects either halfslab of $R_k$.
\end{restatable}

\begin{proof}
Let $R_i$ and $R_j$ be two axis-aligned rectangular uncertainty regions where the top right vertex of $R_i$ is dominated by the bottom left vertex of $R_j$. 
All choices of $p_i \in R_i$  are dominated by the top right vertex of $R_i$, similarly all choices of $p_j \in R_j$ dominate the bottom left vertex of $R_j$ hence via transitivity $p_j$ always dominates $p_i$ which implies that $R_i$ is a negative region. 
If there is no region whose bottom left vertex dominates the top right vertex of $R_i$, then $p_i$ appears on the \pareto of $P$ if all regions have their point lie on the bottom left vertex and $p_i$ lies on the top right vertex of $R_i$. Hence $R_i$ is then not negative. 

If $R_i$ is non-negative, and there exists a region $R_k$ that contains $R_i$ in its horizontal or vertical halfslab then $R_i$ cannot be positive since if $p_k$ is placed on the top right vertex of $R_k$ and $p_i$ on the bottom left vertex, $p_k$ must dominate $p_i$. 

Suppose that $R_i$ is not positive and not negative. Then per definition there exists a point placement of $p_i$, and another true point $p_l$, such that $p_l$  dominates $p_i$. In this case, $p_l$ also  dominates the bottom left vertex of $R_i$, yet the uncertainty region $R_l$ cannot be entirely contained in the quadrant that dominates the top right vertex of $R_i$, else $R_i$ is negative. Hence $R_l$ must have a halfslab that intersects $R_i$ which proves the lemma.
\end{proof}

\noindent
Evans and Sember~\cite{evans2011possible} and  Nagai~\etal~\cite{nagai1999convex} study convex hulls and Pareto fronts
of imprecise points.
They note that for a set of pairwise-disjoint convex regions $\RR$, there is a connected area of negative points.
They call this area the {\it guaranteed dominated region}. 
We refer to the boundary of the guaranteed dominated region as the \emph{guaranteed boundary} $\BR$. We note that for Pareto fronts, the guaranteed boundary is the Pareto front of the bottom left vertices in $\RR$.
Intuitively, discovering the exact location of a point below $\BR$ does not provide additional useful information, only discovering that a point lies below $\BR$ does.

\begin{restatable}{lem}{intersection}
\label{lemma:intersection}
Let $\RR$ be a set of pairwise disjoint non-negative rectangles.
The intersection of a region $R_i \in \RR$ with $\BR$ is a staircase with no top right vertex. 
\end{restatable}

\begin{proof}
Per definition, non-negative regions have a top right vertex that lies above $\BR$. Their bottom left vertex lies either on $\BR$, or below $\BR$ (since $\BR$ is the  \pareto of all bottom left vertices). Hence the closure of each uncertainty region intersects $\BR$.
The intersection between a connected staircase and an axis-aligned rectangular region is always a connected staircase. Each top vertex of $\BR$ corresponds to a bottom left vertex of a region in $\RR$. Each $R_i$ cannot cannot contain such a top vertex since regions are pairwise disjoint.
\end{proof}

\noindent
We formalise the above intuition by defining a procedure $\mathit{Trunc}$.
Given an {\em original} set $\RRo$ of $n^{\text{\tiny{\textregistered}}}$ pairwise disjoint axis-aligned rectangles, $\mathit{Trunc}(\RRo)$ returns a \emph{truncated set} $\RRt$ where some regions may be \emph{flagged} (marked with a boolean).
Refer to Figure~\ref{fig:truncation}.
Specifically, each negative region in $\RRo$ gets removed, each potential region $R_i$, whose bottom left vertex is below $\mathbb{B}_{\RRo}$, gets \emph{flagged} and replaced by the part of $R_i$  
above $\mathbb{B}_{\RRo}$.
By Lemma~\ref{lemma:intersection} this results in a rectangular area. 
All remaining regions are rectangles which touch $\mathbb{B}_{\RRo}$.
Since they are also disjoint, their intersections with $\mathbb{B}_{\RRo}$ induce a well-defined order, and 
$\mathit{Trunc}$ re-indexes the remaining regions according to top left to bottom right ordering of their bottom left vertices.
We obtain a set $\RRt = (R_1, R_2, \ldots R_{n^{\text{\LeftScissors}} }) = \mathit{Trunc}(\RRo)$ with $n^{\text{\LeftScissors}} \le n^{\text{\tiny{\textregistered}}}$.
Observe that $\mathbb{B}_{\RRo} = \mathbb{B}_{\RRt}$.
We say $\RRt$ is a \emph{truncated set} if it is the result of a truncation of some set $\RRo$.

 \subparagraph{Dependency graphs.}
Given a truncated set $\RR=\RRt$, we define a \emph{(directed) dependency graph} denoted by $G(\RR)$ as follows.
The nodes of the graph correspond to the regions in $\RR$.
We have two types of directed edges which we refer to as horizontal and vertical arrows. 
A region $R_i$ has a {\em vertical arrow} to $R_j$ if $R_j$ {\em succeeds} $R_i$ and is vertically visible from $R_i$ (that is, there exists a vertical segment connecting $R_i$ and $R_j$ that does not intersected any other region in $\RR$).
A region $R_i$ has a {\em horizontal arrow} to $R_j$ if $R_j$ {\em precedes} $R_i$ and is horizontally visible from $R_i$.
Refer to Figure~\ref{fig:arrowdirection}.
Observe that, if $\RR$ is a truncated set, any point region $p \in \RR$ has no outgoing arrows, since after truncation the halfslabs of $p$ do not intersect the interior of any rectangle in $\RR$.
We note an important property of the dependency graph:

\begin{figure}[b]
    \centering
    \includegraphics{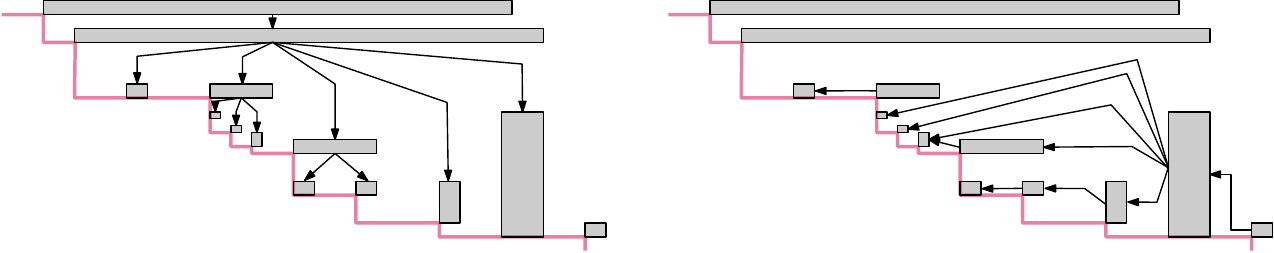}
    \caption{A truncated set and its horizontal and vertical arrows.}
    \label{fig:arrowdirection}
    \vspace{0cm}
\end{figure}

\begin{restatable}{lem}{independence}
\label{lemma:independence}
Let $R_i \in \RR$ such that $R_i$ is a source in $G(\RR)$.
Then all $R_l \in \RR$ with $i < l$ cannot have an incoming dependency arrow from a region $R_k$ with $k < i$ and vice versa.
\end{restatable}

\begin{proof}
Consider such regions $R_k$, $R_i$ and $R_l$.  Per the ordering of $\RR$, the bottom left vertex of $R_k$ lies left and above the bottom left vertex of $R_i$. Per definition, $R_k$ can only have a vertical arrow to $R_l$. The region $R_k$ has a vertical arrow to $R_l$ only if its bottom facet lies above $R_l$. However, then either its bottom facet intersects $R_i$ (contradicting the assumption that the regions are pairwise disjoint) or it lies above $R_i$ (contradicting the assumption that $R_i$ is a source node in $G(\RR)$). The argument for arrows from $R_l$ to $R_k$ is symmetrical.
\end{proof}

\noindent

\begin{corollary}
\label{cor:subproblem}
Let $\RR$ be a truncated set and let $R_i$ and $R_j$ be source nodes in $G(\RR)$. 
There is no region in $\RR \backslash [R_i, R_j]$ that has a directed path in $G(\RR)$ to any region in $[R_i, R_j]$. 
\end{corollary}

\subparagraph{The Pareto cost function.}
We show that for any set $\RRo$, we can construct the \pareto of the underlying point set using only $\RRt = \mathit{Trunc} (\RRo)$.
To show that we can use $\RRt$ to construct $\Xi^*$ in uncertainty-region optimal time, we define  the \emph{Pareto cost function} denoted by $\CP(\RRt, P)$. 
In Section~\ref{sec:lowerbound} we show that $\CP(\RRt, P)$ is the uncertainty-region lower bound for constructing $\Xi^*$ and in Section~\ref{sec:upperbound} we show that this lower bound is tight.

Before we can define the Pareto cost function, we define additional concepts (Figure~\ref{fig:slabregions}). 
By $C$ we denote the unspecified cost for a retrieval.
Whenever we write $\log$ we refer to the logarithm base 2. Let $\RR = \RRt$ be a truncated set.
For all regions $R_i \in \RR$, we denote by $V_i$ the subset of $[R_i, R_n]$ that is vertically visible from $R_i$ (including $R_i$ itself) and by $H_i$ the subset of $[R_1, R_i]$ that is horizontally visible from $R_i$ (including $R_i$ itself). 
Given $P$, we denote by $V_i(P) \subseteq V_i$: the union of $ \{ R_i \}$ with the subset of $V_i$ of regions that are dominated by a point $p_j$ with $j \le i$. The set $H_i(P)$ is defined symmetrically taking points $p_j$ with $i \le j$.

Intuitively, the truncation operator represents the foresight about the  \pareto of $P$.
Now, given a truncated set $\RR$ and $P$ we construct a set $\RP \subset \RR$ that intuitively represents which regions of $\RR$ were geometrically interesting in hindsight. 
Consider for a given $P$, all regions that are intersected by the \pareto of $P$. Let $R_j$ be such a region, then given the \pareto of $P \backslash \{p_j\}$, $R_j$ covers some area above this \pareto. Hence, the point $p_j$ could be part of the \pareto of $P$ if it lies in this area. 
Intuitively, all regions intersected by the \pareto of $P$ are hereby suitable for further inspection; however, if the regions are \emph{positive} regions this further inspection might not be required to construct $\Xi^*$.
Similarly, if the region $R_j$ lies above the Pareto front of the points $P \backslash \{p_j\}$, the point $p_j$ cannot be dominated by a point in $P \backslash \{p_j\}$ and hence we can conclude it lies on the \pareto of $P$ without further inspection. 
This is why we define $\RP$ 
as the subset of $\RR$ where each region $R_i \in \RP$ is intersected by the \pareto of $P$ \emph{and} one of three conditions holds:
\begin {enumerate}
  \item $R_i$ is flagged;
  \item $R_i$ intersects and edge $e$ with endpoint $p_j \in P$ and $i \neq j$; and/or 
  \item $R_i$ is not a sink in $G(\RR)$.
\end {enumerate}
\noindent
We define the \emph{Pareto cost function} as: 
$
\CP(\RR, P) = \sum_{R_i \in \RP} C + \log |V_i(P)| + \log |H_i(P)|.
$

\begin{figure}[b]
    \centering
    \includegraphics{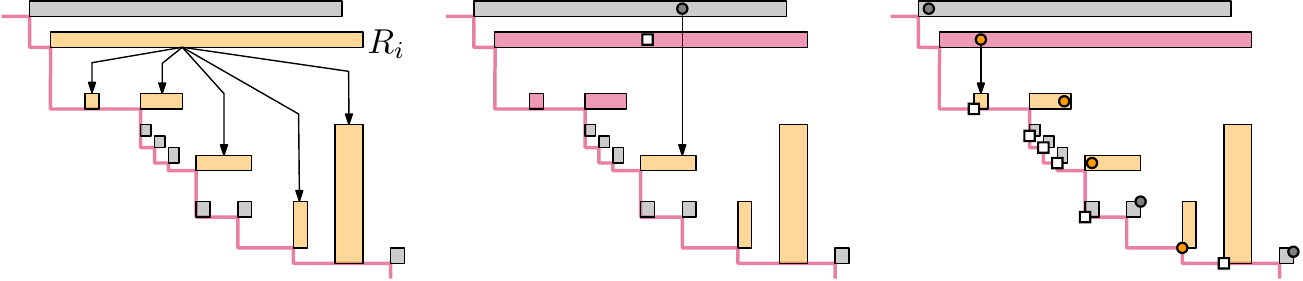}
    \caption{A region $R_i$ and the set $V_i$ in orange. Middle: for a given set of points, the set $V_i(P)$ is shown in red. Right: the set $V_i(P)$ changes for different $P$, but always includes $R_i$. }
    \label{fig:slabregions}
    \vspace{0cm}
\end{figure}

\section{Lower bounds}
\label{sec:lowerbound}

One is free to compute any auxiliary $\Xi$ in the preprocessing phase, in order to reconstruct a structure $\Xi^*$, isomorphic to the \pareto, as efficiently as possible.
There exists a choice of input $\RRo$ where all regions are positive: namely whenever $\RRo = \RRt = Trunc(\RRo)$ and $G(\RRo)$ is a graph with no edges. 
In this case, for every choice of $P$ that respects $\RRo$, the \pareto of $P$ is isomorphic to $\mathbb{B}_{\RRo}$ hence it is possible to construct $\Xi^*$ in the preprocessing phase. If $\RRo$ has $m$ elements, 
constructing $\mathbb{B}_{\RRo}$ has a well-known $O(m \log m)$ worst case lower bound. \\

In the reconstruction phase an algorithm can use \emph{any} auxiliary structure $\Xi$ to aid its computation.
In the remainder of this section we consider \emph{any} truncated set $\RR = \RRt = Trunc(\RRo)$ of $n$ elements, together with \emph{any} auxiliary datastructure. We provide an information-theoretical lower bound, which depends on $\RR$ and $P$, for both the number of RAM instructions and disk retrievals required to construct $\Xi^*$ regardless of $\Xi$.

\subsection{A lower bound for disk retrievals}
\label{sub:retrievals}

Bruce~\etal
study in their paper the reconstruction of the \pareto of $P$ in a variant of (what would later be) the preprocessing model with implicit representation. Bruce~\etal present an iterative retrieval strategy that is instance optimal. Their strategy performs at most three times more retrievals than any algorithm must use to discover the \pareto of $P$ and they prove that this factor-3 redundancy is the best anyone can do. Their strategy describes the regions that must be considered in a geometric sense, not an algorithmic sense. That is, at each iteration they can identify a triplet of regions to query. But they have no algorithmic procedure to identify these three regions as such, nor a way to beforehand specify which regions should be considered. 
In their model this is justifiable as they assume that the retrieval cost $C$ vastly dominates any RAM instructions and hence identifying the triple each iteration is trivial. 
In this paper, we drop the assumption that $C$ is enormous and are interested in a retrieval strategy which not only minimizes the number of retrievals, but which can also elect which points to retrieve efficiently.

We note that the query strategy of Bruce~\etal produces a result of the same quality as the lemma below and naturally, our proofs share some elements which we fully wish to attribute to the work of \cite{bruce2005efficient}. The novelty in our result is that for each pair $(\RR, P)$ we are able to characterize the regions which require a disk retrieval using $\RP$. Which will help us in the reconstruction phase, when we want to identify these regions efficiently.

\begin{figure}[b]
    \centering
    \includegraphics{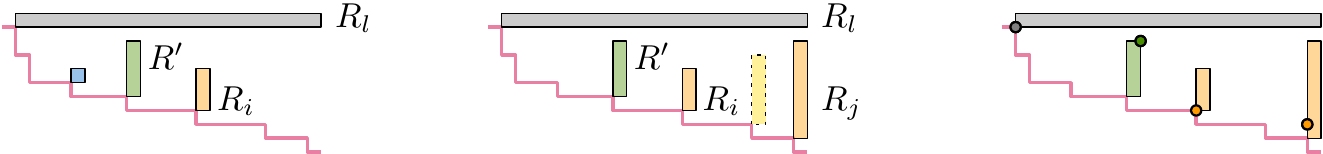}
    \caption{Left: The region $R_l$ charges the blue region and $R_i$ the green. Middle: for $R_j$, either  $R_i \in H_j$ or there is another region (yellow) with higher index in $H_j$. }
    \label{fig:lowerboundcharge}
    \vspace{0cm}
\end{figure}
\begin{restatable}{lem}{retrieveLowerbound}
\label{lemma:retrievelowerbound}
Let $\RR$ be a truncated set and let $P$ be any point set that respects $\RR$. Any algorithm that constructs $\Xi^*$ of $P$ must perform at least $\frac13|\RP|$ retrievals.
\end{restatable}
\begin{proof}
Let $R_i \in \RP$. Per definition, $R_i$ is not dominated by a point in $P$. Hence given $P \backslash p_i$, there exists a choice of $p_i$ such that $p_i$ appears on the \pareto of $P$.
Any algorithm $A$ \emph{must} spend a disk retrieval on $p_i$, if there \emph{also} exists a choice of $p_i$ such that it does not appear on the \pareto, given $P \backslash p_i$.
We consider the three cases for when $R_i \in \RP$:

Let $R_i$ be flagged. Then there exists a choice of $p_i$ such that $p_i$ lies below $\BR$ and hence does not appear on the \pareto of $P$.
Else let $R_i$ be intersected by an edge that has as an endpoint a point $p_j$ with $j \neq i$. Then $e$ is either a vertical edge whose top vertex is $p_j$ or a horizontal edge whose right vertex is $p_j$. In both cases, there exists a choice of $p_i$ for which it does not appear on the \pareto of $P$ since it would be dominated by $p_j$ (this is achieved by placing $p_i$ left of the vertical edge, or below the horizontal edge).
Lastly let neither first two cases apply and $R_i$ have at least one outgoing edge in $G(\RR)$. 
Then there is at least one region $R' \in H_i \cup V_i$, the argument for this case is illustrated by Figure~\ref{fig:lowerboundcharge}.
Denote by $R'$ a region in $H_i$ (the case for $V_i$ is symmetrical). 
Moreover, let $R'$ be the region in $H_i$ with the highest index. 
We `charge' the region $R'$ one disk retrieval. First we show that each region in $\RR$ gets charged at most twice, then we show this charge is justified. 

Suppose that $R'$ gets charged by two regions $R_i$, $R_j$ with $R' \in H_i$ and $R' \in H_j$ (the argument for when $R'$ lies in two vertical halfslabs is symmetrical) and let $i < j$.
If $R'$ lies in $H_i$ and $H_j$, then $R_i$ must lie in the horizontal halfslab of $R_j$, which contradicts the assumption that $R'$ was the region in $H_j$ with the highest index (see Figure~\ref{fig:lowerboundcharge}, middle).

Second we show that this charge is justified.
Consider $R'$ and the two regions $R_i$ and $R_l$ ($l < i$) that charge $R'$ and all points in $P \backslash \{p', p_i, p_l \}$. 
Since case (2) does not apply to $R_i$ and $R_l$, there is no point $p \in P \backslash \{p_i, p_l \}$ whose horizontal or vertical halfslab intersects $R_i$ or $R_l$, thus no point in $P \backslash \{ p_i, p_l \}$ can dominate $R'$, $R_i$ or $R_l$.
This implies that regardless of all other points, there a choice for $p_i, p_l, p'$ where all three points appear on the \pareto of $P$ (the point placement where $p_i$ and $p_l$ appear on the bottom left vertex of their respective regions and $R'$ appears on the top right vertex).
However, there also exists a choice where $p'$ is dominated by $p_l$ or $p_i$.
Any algorithm must therefore consider at least $p', p_i$ or $p_l$ in order to find out and this is why the charge is justified. 
\end{proof}

\subsection{A lower bound on RAM instructions}
In Section~\ref{sec:optimality} we defined the uncertainty-region lower bound. 
By an information-theoretical lower bound (algebraic decision tree or entropy \cite{afshani2017instance,cardinal2019information}), 
we have, for any $\RR$, that the Uncertainty-region lower bound is at least $\log L(\RR)$, where $L(\RR)$ is the number of combinatorially different Pareto fronts of point sets 
that respect $\RR$.
We prove the following:

\begin{restatable}{lem}{realLower}
\label{lemma:realLower}
Let $\RR$ be a truncated set and $P$ be any point set that respects $\RR$. Then
\[
\sum_{R_i \in \RP} \log |V_i(P)| + \log |H_i(P)| \le 2 \cdot \log L(\RR)\,.
\]
\end{restatable}

\begin{proof}
We show that $\sum_{R_i \in \RP} \log |V_i(P)| \le \log L(\RR)$. By a symmetric argument we have $\sum_{R_i \in \RP} \log |H_i(P)| \le \log L(\RR)$ and the lemma follows.
Consider for a fixed set $P$ all regions $R_i \in \RP$ for which $|V_i(P)| \ge 2$ (recall that $R_i \in V_i(P)$) and sort them from lowest index to highest.
 For ease of exposition  we denote these regions as $(R^1, R^2, \ldots, R^m)$.
We create $m$ different, pairwise disjoint vertical slabs as follows: the first slab is bound by the left facets of $R^1$ and $R^2$, the second by facets of $R^2$ and $R^3$ and the $m$'th slab is a halfplane (Figure~\ref{fig:verticalslabs}).
In the degenerate case that a slab has width $0$ (this can occur, when after truncation regions can have left vertices that share a coordinate) we give it width $\varepsilon$.

Let $R_i = R^1$ and $R_j = R^2$. For all regions $R_k \in V_i(P)$, per definition $i \le k < j$. Each of these truncated regions has thus a bottom left endpoint that lies left of the bottom left vertex of $R_j$ and right of the bottom vertex of $R_i$ which implies that their bottom left vertex lies in the first vertical slab.
The result of this observation is, that given $\RR$, there are at least $|V_i(P)|$ combinatorially different Pareto fronts contained within the first vertical slab. These Pareto fronts are obtained by placing the points of the regions in $V_i(P) \backslash R_i$ on their respective bottom left endpoints, and by letting $p_i$ dominate any prefix of these points.

Let $R_j = R^2$ and $R_k = R^3$. Via the same argument each region in $V_j(P)$ has its bottom endpoint in the second vertical halfslab. Hence with the same argument as above, there are at least $|V_j(P)|$ combinatorially different Pareto fronts contained within the second halfslab. Moreover, we created $|V_i(P)|$ different combinatorial outcomes by placing only points in the first vertical halfslab, using only points preceding $p_j$. This means that these combinations can be generated, whilst no point preceding $p_j$ dominates any point following $p_j$. This implies that the total number of combinatorially different Pareto fronts contained in both the first and second halfslab is $|V_i(P)| \cdot |V_j(P)|$.
By applying this argument recursively it follows that:
$
\prod_{R_i \in \RP} |V_i(P)| \le L(\RR),
$
which concludes the proof.
\end{proof}

\noindent
Given Lemma~\ref{lemma:retrievelowerbound} and Lemma~\ref{lemma:realLower} we can immediately conclude the following:

\begin{figure}[t]
    \centering
    \includegraphics{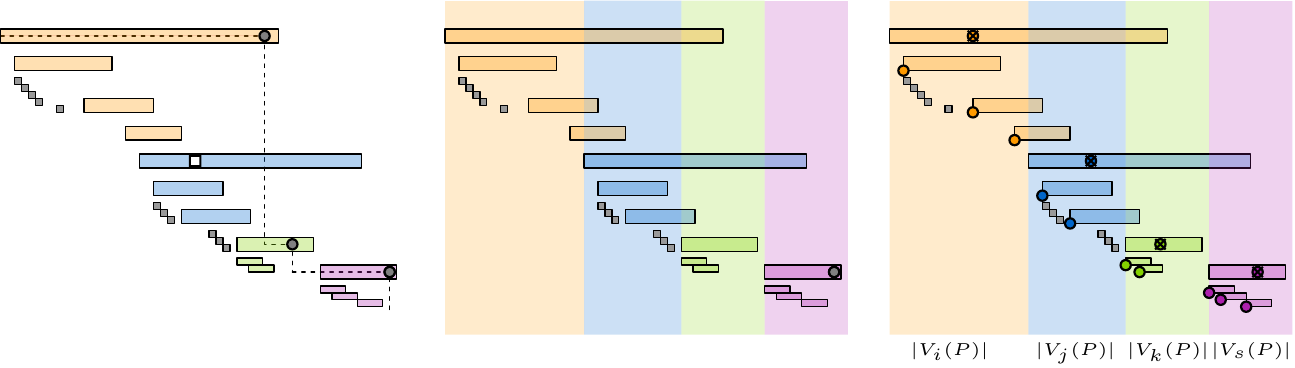}
    \caption{Left: A pair $(\RR, P)$ such that the grey points form the \pareto. Given the \pareto, we can extract $V_i(P)$ for each $i$. Middle: based on the sets $V_i(P)$, we create vertical slabs irrespective of the original points $P$. Right: In each vertical slab, we can create $V_i(P)$ combinatorially distinct (partial) Pareto fronts using only points in the vertical slab. }
    \label{fig:verticalslabs}
    \vspace{0cm}
\end{figure}

\begin{restatable}{theorem}{lowerbounds}
\label{thm:lowerbounds}
Let $\RR$ be a truncated set and $P$ be \emph{any} set that respects $\RR$. Then $\CP(\RR, P)$ is fewer than three times the uncertainty-region lower bound of $\RR$. 
\end{restatable}
\noindent
We wish to briefly note that for each $i$, $V_i(P)$ and $H_i(P)$ have at most $n$ elements and thus by Lemma~\ref{lemma:retrievelowerbound}, $\CP(\RR, P)$ is a factor $\log n$ removed from the instance lower bound. 

\section{Reconstructing a \pareto}
\label{sec:upperbound}

Theorem~\ref{thm:lowerbounds} gives an uncertainty-region lower bound for any truncated set $\RR$. In this section, we show that this lower bound is tight. 
To that end, we first define additional geometric concepts.
First, we introduce the notion of \emph{canonical} rectangles.
Then we define the notion of \emph\subproblems. 
Finally, we show how to use the \subproblems of a canonical set to quickly select only regions which lie in $\RP$.
We wish to emphasise that in the reconstruction phase we have \emph{implicit} access to the point set $P$, meaning that for each region $R_i$, we can request $p_i$ in $O(C)$ time. Thus reading all points in $P$ takes $\Omega(n C)$ time, which we aim to avoid.

\subsection{Geometric preliminaries for reconstruction}
\label{sub:reconstPrelims}
Let $\RR$ be a truncated set of $n$ regions and let $P$ respect $\RR$.
Denote by $\Vnext{i}$ the region strictly right of the vertical slab of $R_i$ with the lowest index; $\Hprev{i}$ is defined symmetrically using the highest index (refer to Figure~\ref{fig:pointerstructure}).
For each $i$, let $p_i^\mathit{xMax}$ (respectively $p_i^\mathit{yMax}$) be the point in $P$ with maximal $x$-coordinate ($y$-coordinate) among points $p_k$ with $k \le i$ (with $k \ge i$).
Throughout this section, we denote by $f_i(P)$ the region succeeding $R_i$ with the lowest index that is not dominated by a point $p_k$ with $k \le i$. 
The region $g_i(P)$ is the region preceding $R_i$ with highest index not dominated by a point $p_k$ with $k \ge i$.

Let $R_i \in \RR$ be both a source and sink in $G(\RR)$.
By Lemma~\ref{lemma:independence}, $p_i$ appears on the \pareto and connects the \pareto of $[p_1, p_{i-1}]$ and $[p_{i+1}, p_n]$.
Thus, we can split the problem of computing the \pareto of $P$ into two, and solve each half independently.
We say that a truncated set $\RR$ is \emph{culled} if $G(\RR)$ contains no region that is both a source and a sink.
Let $[R_i, R_j]$ be a sequence of sinks in $G(\RR)$, and $R^*$ be the smallest rectangle that contains $R_i$ and $R_j$.
Note that $R^*$ is disjoint from regions in $\RR \backslash [R_i, R_j]$ and contains all $[R_i, R_j]$. We can use $R^*$ to capture a ``streak'' of points which do, or do not, appear on the \pareto:

\begin{restatable}{lem}{contiguous}
\label {lem:contiguous}
Let $[R_i, R_j]$ be a sequence of sinks in $G(\RR)$. If there is no $p_k \in P$ preceding $p_i$ that dominates $p_i$ then there is no point preceding $p_i$ that dominates any point in $[p_i, p_j]$. If some $p_k$ preceding $p_i$ dominates $p_j$, then $p_k$ dominates all points in $[p_i, p_j]$.
Similar statements hold for $p_k$ succeeding $p_j$.
\end{restatable}
\begin{proof}
Any  $p_k$ that dominates any point $p_s$ with $s \in \langle i , j \rangle$, but not $p_i$ or $p_j$ itself must lie in the interior of $R^*$, but $R^*$ contains only points whose regions are sinks in $G(\RR)$. This contradiction implies all claims of the lemma.
\end{proof}
\noindent
This lemma implies that if both $p_i$ and $p_j$ are not dominated by other points in $P$ then all the points in $[p_i, p_j]$ appear on the \pareto of $P$ as a contiguous subsequence, and all regions $R_k \in [R_i, R_j]$ are not part of $\RP$.  Theorem~\ref{thm:lowerbounds} states we cannot ``afford'' to spend any disk retrievals on $(p_i, p_{i+1}, \ldots, p_j)$. Instead, we should add a pre-stored chain referencing $[p_i, p_j]$ to $\Xi^*$ in constant time. 
This is why for any maximal sequence of sinks $[R_i, R_j]$ in a truncated and culled set $\RR$, we define their \emph{compound region} $R^*$ and we replace $[R_i, R_j]$ in $\RR$ with $R^*$ (refer to Figure~\ref{fig:transformation} (right)).
Let $\RRc$ be the resulting set of regions.
The region $R^*$ is a sink in $G(\RRc)$ and a region $R$ has an outgoing arrow to $R^*$ in $G(\RRc)$ if and only if it had an outgoing arrow in $G(\RRt)$ to at least one region in $[R_i, R_j]$.
Since $R^*$ is just another rectangle disjoint from all other rectangles in $\RR^\mathit{comp}$, the definition of \emph{truncated} and \emph{culled} still applies to $\RR^\mathit{comp}$.
We say a set $\RRc$ is a \emph{canonical set} if it is truncated, culled, and if there are no two consecutive regions that are sinks in $G(\RRc)$.
In the remainder, we assume $\RR$ is a truncated set and $R^0 = \RRc$ is its respective canonical set as the reconstruction input.

\begin{figure}[tb]
    \centering
    \includegraphics{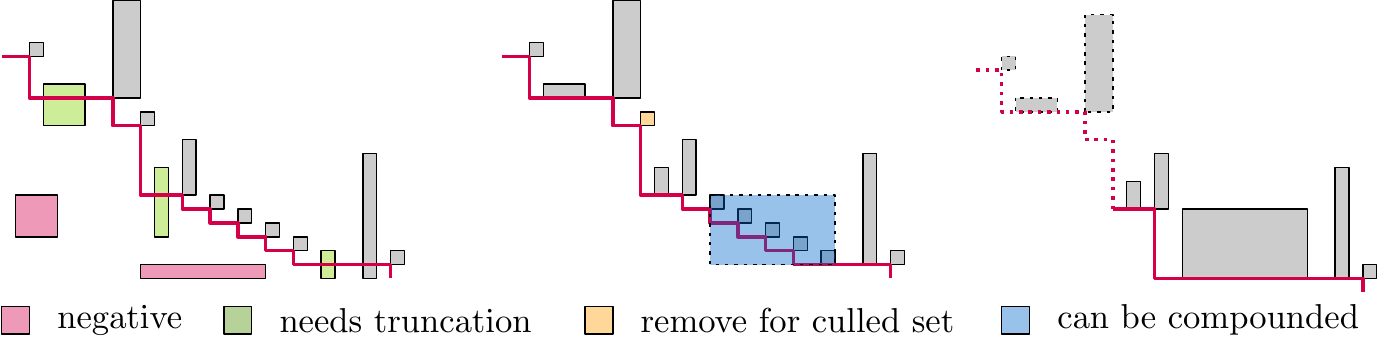}
    \caption{Left: $\RRo$ with $\mathbb{B}_{\RRo}$ in red. Middle: the set of regions after truncation. The yellow region is a source and a sink, it splits the problem into two. Right: The canonical set. }
    \label{fig:transformation}
    \vspace{0cm}
\end{figure}

\subparagraph*{Subproblems.}
Let $\RR$ be a \emph{truncated}  set. 
We say two indices $i < j$ form a \emph{\subproblem} with respect to a dependency graph $G(\RR)$ if $R_i$ and $R_j$ are sources in $G(\RR)$ and if there does not exist a region $R_k$ with $i < k < j$ that is also a source. With slight abuse of notation, we say that $[R_i, R_j]$ is a \subproblem of $G(\RR)$.
At later stages we will consider some altered dependency graph $G(\RR')$ and will refer to \subproblems $[R_l, R_m]$ of $G(\RR')$.

\subparagraph*{The algorithm sketch.}
The core of our algorithm is rather straightforward:
it is an iterative strategy, where at each iteration $t$ we have an (implicitly truncated) set $\RR^t$ and a queue of \subproblems of $G(\RR^t)$. Each iteration, we dequeue a \subproblem $[R_i, R_j]$ of $G(\RR^t)$, retrieve $p_i, p_j$ to replace $R_i$ and $R_j$ and (implicitly) re-truncate.
We maintain the following invariant: 
\begin{invariant}
\label{inv:pointer}
For each iteration, when we consider a \subproblem $[R_i, R_j]$ 
we have a pointer to the region $R$ which stores $p_{i-1}^\mathit{xMax}$ and the region $R'$ which stores $p_{j+1}^\mathit{yMax}$.
\end{invariant}
\noindent
Observe that for all \subproblems $[R_i, R_j]$ of $G(\RR^0 = \RRc)$, the point $p_{i-1}^\mathit{xMax} = p_{i-1}$ and $p_{j+1}^\mathit{yMax} = p_{j+1}$. 
We sketch Algorithm~\ref{algo:upperbound}. 
We want to prove that its runtime matches the value $\CP(\RR, P)$ of Theorem~\ref{thm:lowerbounds}. This would trivially be true, if  for each \subproblem $[R_i,R_j]$ of $G(\RR^t)$, $R_i, R_j \in \RP$.
Unfortunately that is not always the case, and thus we resort to a more involved argument to prove the following theorem. In the remainder of this section, we show that the algorithm's running time is $O(A(\RR, \RRc, P))$. 

\begin{restatable}{theorem}{runtime}
\label{thm:runtime}
Let $\RR$ be a truncated set and let $\RRc$ be its respective canonical set, $\Xi$ be built on $\RRc$ and Algorithm~\ref{algo:upperbound} run on $\RRc$ as input.  Let Algorithm~\ref{algo:upperbound} consider for each iteration $t$, a \subproblem $[R_{i(t)}, R_{j(t)}]$ with $i(t) < j(t) - 1$. Let $\RR^{A1}( \RRc, P) = \bigcup_t \{R_{i(t)}, R_{j(t)}\}$. Let $V_i(P)$ and $H_i(P)$ refer to subsets of $\RR$, not $\RRc$.
Then:
\[
A(\RR, \RRc, P) = \sum_{R_i \in \RR^{A1}(\RRc, P) } \left( \frac{1}{2}C + \log |V_i(P)| + \log |H_i(P)| \right) \le \CP(\RR, P)\,. 
\]
\end{restatable}

\begin{algorithm}[H]
\label{algo:upperbound}
\SetAlgoLined
\KwResult{The pointer structure \textbf{$\Xi^*$}.  \hspace{5,6cm} (Runtime)}
 $Q \gets$ \subproblems ($G(\RR^0)$) \hspace{6,4cm} (Preprocessing) \\
 \While{$Q \neq \emptyset$}{
 $[R_i, R_j] \gets $ Q.DeQueue() \hspace{6,9cm} ($O(1)$) \\
  $p_i, p_j \gets$ Retrieve($R_i$, $R_j$)   \hspace{6,3cm} ($2C + O(1)$) \\
  $p_i^\mathit{xMax}, p_j^\mathit{yMax} \gets$ Compare(($p_i$,\, $p_{i-1}^\mathit{xMax}$), ($p_j$, $p_{j+1}^\mathit{yMax}$))  \hspace{2,41cm} ($2C +O(1)$) \\
  \uIf{$p_i$ not dominated by $_{i}^\mathit{xMax},\, p_{j}^\mathit{yMax}$}{
  \textbf{$\Xi^*$}.Append($p_{i}$ after $p_{i-1}^\mathit{xMax}$) \hspace{5,9 cm} ($O(1)$) \\
  }
  \uIf{$p_j$ not dominated by $p_{i}^\mathit{xMax},\, p_{j}^\mathit{yMax}$}{
  \textbf{$\Xi^*$}.Append($p_{j+1}^\mathit{yMax}$ after $p_{j}$) \hspace{5,9 cm} ($O(1)$) \\
  }
  $f_i(P) \gets$ gallopingSearch($p_{i}^\mathit{xMax},\, V_i$) \hspace{4,5 cm} ($O(\log |V_i(P)|)$) \\
  $g_j(P) \gets$ gallopingSearch($p_{j}^\mathit{yMax},\, H_j$) \hspace{4.3 cm} ($O(\log |H_j(P)|)$) \\
  $R^{t+1} \gets$ ImplicitTruncate($R^t - R_i-  R_j + p_i + p_j$) \hspace{3 cm} ($O(1)$) \\
  DetermineSubproblems($\RR^{t+1}$,\, $f_i(P)$,\, $g_j(P)$) \hspace{3,85 cm} ($O(1)$) \\
  \ForEach{\subproblem $[R_c, R_d]$ of $G(\RR^{t+1} \cap [R_i = p_i,\, R_j = p_j])$} {
  Q.Queue($[R_c, R_d]$)\hspace{5 cm} ($O(1)$, charged to $[R_c, R_d]$)
  }
 }
 \caption{Algorithm sketch, assuming $\RR^0$ is canonical.}
\end{algorithm}

\subparagraph{Proving Theorem~\ref{thm:runtime}.}

This theorem describes an intuitive ``runtime allowance'' that Algorithm~\ref{algo:upperbound} has. 
We first prove 3 Lemmas about \subproblems encountered by Algoritm~\ref{algo:upperbound}.

\begin{restatable}{lem}{consideration}
\label{lemma:consideration}
Let $\RR$ be a canonical set and $R_i \in \RR$.
Algorithm~\ref{algo:upperbound} encounters a \subproblem $[R_i, \cdot]$ or $[\cdot, R_i]$ if and only if $R_i$ is intersected by the \pareto of $P$.
\end{restatable}

\begin{proof}
The region $R_i$ is not intersected by the \pareto of $P$ if and only if $R_i$ is dominated by a point $p_j \in P$.  Let $p_j$ appear on the \pareto of $P$ (via transitivity of domination, we can always obtain such a $p_j$). 
The iterative procedure must consider $p_j$ before $p_i$ since $R_j$ prevents $R_i$ from being a source in the dependency graph. But when $R_j$ is considered, $R_i$ is truncated. 
The graph must always have at least one source. Thus, since $R_i$ will never be removed after truncation, it must eventually become a source.
\end{proof}

\begin{lemma}
\label{lemma:eitheror}
Let $\RR^0$ be a canonical set. Algorithm~\ref{algo:upperbound} encounters only \subproblems $[R_i, R_j]$ where  either: $j = i+1$ or $R_i \in \RP$ or $R_j \in \RP$, and $R_i \not \in \RP$ if and only if $|V_i(P)| = |H_i(P)| = 1$ (the same holds for $R_j$).
\end{lemma}

\begin{proof}
If $\RR^0$ is a canonical set, then there cannot by any \subproblem $[R_i, R_j]$ of $G(\RR^0)$ where $R_i$ and $R_j$ are both sinks in $G(\RR^0)$.
As a consequence, for each $[R_i, R_j]$ either $R_i \in \RP$ or $R_j \in RP$ and $R_i \not \in \RP$ implies $V_i(P) = H_i(P) = \{ R_i \}$.

In later iterations, we cannot immediately guarantee that $\RR^t$ is canonical, and the allowance for spending computation time is hence lost. 
Via Lemma~\ref{lemma:consideration} we know that $R_i$ and $R_j$ are both intersected by the \pareto of $P$. Thus, the regions $R_i, R_j \not \in \RP$ implies that $R_i$ and $R_j$ are both sinks in the original graph $G(\RR)$ (as $\RP$ is defined on the original truncated set). Thus $R_i \not \in \RP$ implies $V_i(P) = H_i(P) = \{ R_i \}$. 

What remains is to show that for each \subproblem either $R_i$ or $R_j$ \emph{does} lie in $\RP$.
Let $i < j - 1$. Then if $R_i$ and $R_j$ are both sinks, then by Lemma~\ref{lemma:independence} the region $R_{i+1}$ or $R_{j-1}$ must also be a source which contradicts the assumption that $[R_i, R_j]$ is a \subproblem.
\end{proof}

\begin{restatable}{lem}{subsequent}
\label{lemma:subsequent}
Let $\RR$ be a canonical set. Algorithm~\ref{algo:upperbound} encounters only \subproblems $[R_i, R_j]$ followed by $[R_i = \{ p_i \}, R_k]$ if $R_k \in \RP$. 
\end{restatable}

\begin{proof}
By the argument of Lemma~\ref{lemma:eitheror}, $R_k$ is intersected by the \pareto of $P$. Moreover after the iteration $t$ where the algorithm considers  $[R_i, R_j]$, the region $R_i$ has no outgoing edges in each iteration $t'$ with $t < t'$. Hence if $[R_i = \{ p_k \}, R_k]$ is a \subproblem, the region $R_k$ has at least one outgoing arrow and thus $R_k \in \RP$.
\end{proof}
\noindent
These three Lemmas imply the following theorem that we later use for a charging scheme: when we relate algorithm runtime to $CP(\RR, P):$

\begin{proof}[Proof of Theorem~\ref{thm:runtime}]
Recall that $\CP(\RR, P)) = \sum_{R_k \in \RP}  C + \log |V_k(P)| + \log |H_k(P)|$.
Let $[R_i, R_j]$ be the first \subproblem considered that has $R_i$ as its left boundary. By Lemma~\ref{lemma:eitheror}, at least $R_i$ or $R_j$ is in $\RP$ hence we charge $\frac{1}{2}C$ time to either the term  $(C + \log |V_i(P)| + \log |H_i(P)|)$ or  $(C + \log |V_j(P)| + \log |H_j(P)|)$ in the sum of $CP(\RR, P)$.
Moreover, $R_i \not \in \RP$ implies $\log |V_i(P)| = \log |V_j(P)| = 0$ hence including these two terms, does not increase the sum's value.
For subsequent \subproblems $[R_i, R_k]$, Lemma~\ref{lemma:subsequent} guarantees that $R_k \in \RP$. Hence the term: $(\frac{1}{2}C + \log |V_k(P)| + \log |H_k(P)|)$ in the sum of $A(\RR, \RRc, P)$ can be charged to the term  $(C + \log |V_k(P)| + \log |H_k(P)|)$ in the sum of $CP(\RR, P)$.  
\end{proof}

\subparagraph{The \subproblem tree.}
Theorem~\ref{thm:runtime} shows that if we are able to execute our described algorithm in the specified running time, then we prove that $\CP(\RR, P)$ is tight and we have obtained an uncertainty-region optimal algorithm. However, in order to achieve this running time, in each iteration we must determine the new \subproblems efficiently.
This is why we define a \subproblem tree on the original dependency graph $G(\RR)$. 
The \subproblem tree, denoted by $T_\RR$, is a range tree on the interval $[1, n] \subset \mathbb{Z}$ (Figure~\ref{fig:dependencytree}). 
The root node of the \subproblem tree stores the interval $[1, n]$.
If $\RR$ is a canonical set, the \subproblems of $\RR$ partition $\RR$, and the root node has a child for each \subproblem $[R_i, R_j]$ where the child stores the interval $[i, j]$ and a pointer to $R_i$ and $R_j$.
We construct the subsequent children as follows: for each node $[i, j]$, we remove all outgoing arrows from $R_i$ and $R_j$ and we create a child node for each \subproblem of $G([R_i, R_j])$ without these arrows.
Note that each node has at least two children: as removing the outgoing arrows from $R_i$ and $R_j$ creates at least one additional source $R_k$ with $k \in \langle i , j \rangle$ and $R_i$ and $R_j$ remain sources in $G([R_i, R_j])$.

\begin{figure}[tb]
    \centering
    \includegraphics{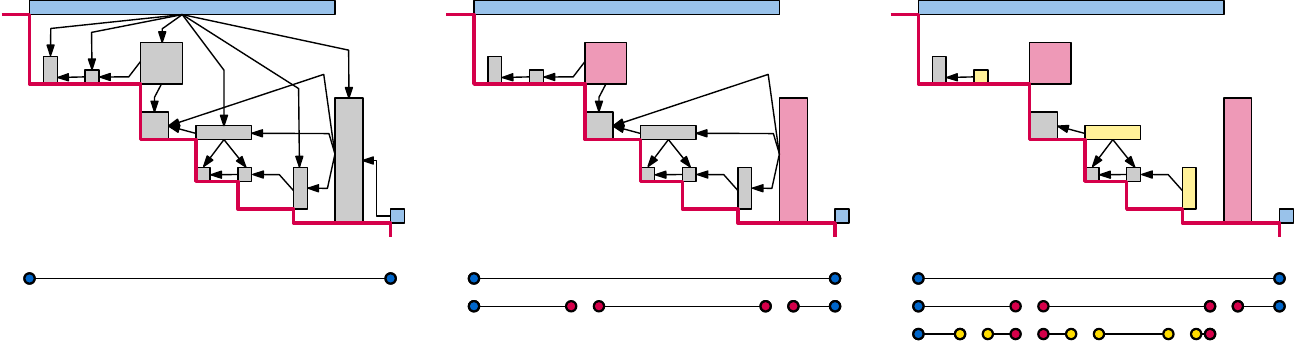}
    \caption{The construction of the \subproblem tree. Left we see a \subproblem of a canonical set with the vertical arrows drawn. In the middle we see the children of this \subproblem with the horizontal arrows drawn. On the right we continued the recursion one additional step. }
    \label{fig:dependencytree}
    \vspace{0cm}
\end{figure}

\subsection{Preprocessing phase}
\label{appx:preprocessing}

Here, we elaborate on the preprocessing procedure.
First, we transform a set $\RRo$ of $m$ axis-aligned pairwise disjoint rectangles into a truncated set $\RRt$ with $n$ elements in $O(m \log m)$ total time.
Next, we construct a canonical set $\RRc$  and the auxiliary datastructure $\Xi$ (which consists of the \subproblem tree $T_{\RRt}$ and some additional pointers) in $O(n \log n)$ time.
Specifically, we define $\Xi$ as follows:

\subparagraph{Defining $\Xi$.}
Given a canonical set $\RRc$, let $\Xi$ consist of $G(\RRc)$ and the tree $T_{\RRc}$ augmented with the following \emph{attributes} stored for every region $R_i \in \RR$ (Figure~\ref{fig:pointerstructure}): 
 \begin{enumerate}
     \item\label{attr1} A binary search tree on $V_i$ and $H_i$ from $G(\RRc)$.
     \item\label{attr2} A pointer to $\Vnext{i}$ and $\Hprev{i}$ in $\RRc$.
     \item\label{attr3} A pointer to the region $R_j$ with highest $j$, such that $R_i \in V_j$ (the \emph{back pointer}) and  a pointer  to the region $R_j$ with lowest index $j$, such that $R_i \in H_j$ (the \emph{forward pointer}).
     \item\label{attr4} A pointer to the highest node in $T_{\RRc}$ that stores an interval $[\cdot, i]$, and a pointer to the highest node in $T_{\RRc}$ that stores an interval $[i, \cdot]$.
     \item\label{attr5} If $R_i$ is a compound region, an array of all the regions compound in $R_i$.
 \end{enumerate}

\begin{figure}[b]
    \centering
    \includegraphics{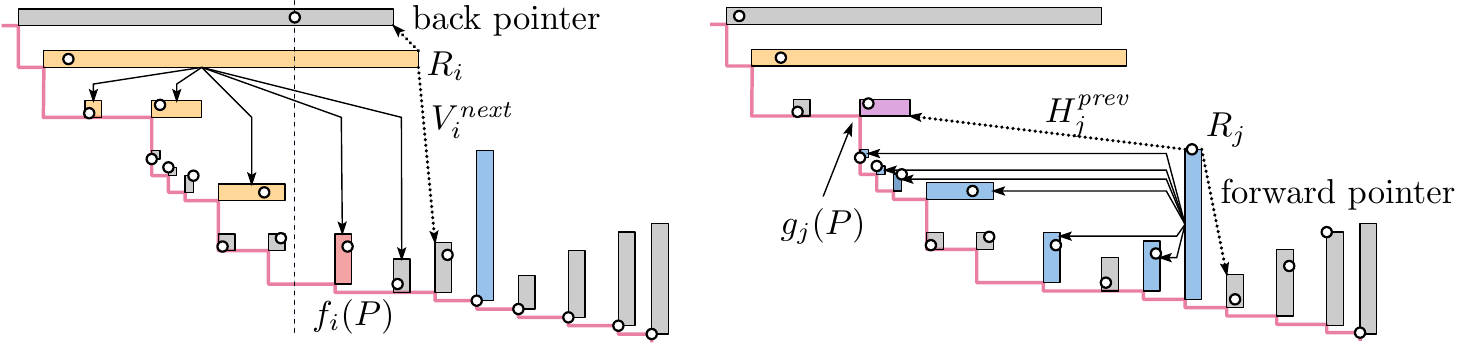}
    \caption{Two choices of $P$ for the same set $\RR$. The sets $R_i(P)$ and $R_j(P)$ are shown in orange and blue respectively. Left: we show $V_i^\mathit{next}$ and the backward pointer and $f_i(P)$. Right: we show $H_j^\mathit{prev}$ and the forward pointer and $g_j(P)$. }
    \label{fig:pointerstructure}
    \vspace{0cm}
\end{figure}

\subparagraph{Creating a truncated set.}
We consider the bottom left vertices of all regions in $\RRo$, construct $\mathbb{B}_{\RRo}$,  together with a range tree on the horizontal edges of $\mathbb{B}_{\RRo}$~\cite{de2008computational} in $O(m \log m)$ time.
For each region $R \in \RRo$ we detect whether $R$ is negative by performing a point location with its top right vertex on the interior of $\mathbb{B}_{\RRo}$; if it is negative then it is discarded. 
If a region $R \in \RRo$ is not negative then by Lemma~\ref{lemma:intersection} we know that
$R \cap \mathbb{B}_{\RRo}$ is a staircase of constant complexity which we compute in logarithmic time using binary search on $\mathbb{B}_{\RRo}$. 
We flag each non-negative $R \in \RRo$ whose interior intersects $\BR$, and store its region after truncation.
This results in a set $\RRt$ of $n$ pairwise disjoint axis-aligned rectangles, which we sort and re-index based on their intersection with $\mathbb{B}_{\RRt}$ in $O(m \log m)$ time and conclude:

\begin{restatable}{lem}{truncation}
\label{lem:truncation}
For any set $\RRo$ of $m$ axis-aligned, pairwise disjoint axis-aligned rectangles we can construct its truncated set $\RRt$ of $n$ rectangles in $O(m \log m)$ time.
\end{restatable}

\noindent
 Recall that for any truncated set $\RRt$ we denote by $H_i$ the set of regions $R_j$ in $\RR$ with $j < i$ which are horizontally visible from $R_i$ and by $V_i$ the set of regions $R_j$ with $j > i$ which are vertically visible from $R_i$. 
 In the remainder of the preprocessing phase, we spend $O(n \log n)$ time to transform $\RRt$ into a canonical set $\RRc$, construct $G(\RRt)$ and $G(\RRc)$ and construct the datastructure $\Xi$.

 \begin{observation}
 \label{obs:decomp}
For any truncated set $\RRt$, a region $R_j \in \RRt$ is vertically visible from a region $R_i \in \RRt$ if and only if there exists a face or edge in the vertical decomposition of $\RR$ which is vertically adjacent to both $R_i$ and $R_j$.
\end{observation}
\noindent
Using Observation~\ref{obs:decomp} we obtain the following through standard Computational Geometry:

\begin{restatable}{lem}{canonicalset}
\label{lem:canonicalset}
For any truncated set $\RRt$ of $n$ axis-aligned, pairwise disjoint rectangles we can construct its canonical set $\RRc$ and $\Xi$ in $O(n \log n)$ time.
\end{restatable}

\begin{proof}
A vertical or horizontal decomposition has a number of faces and edges which is linear in the number of input vertices and can be constructed in $O(n \log n)$ time~\cite{de2008computational}. Given the vertical decomposition of $\RRt$, we can traverse it in linear time to store for each region $R_i$ the set $V_i$. Similarly we can identify and store $H_i$ for each $R_i$, and in $O(n \log n)$  total time we construct a binary search tree on each set $H_i$ and $V_i$ to obtain Attribute 1. For each set $V_i$, we identify $\Vnext{i}$ in logarithmic time by searching by searching for the left-most bottom-left endpoint right of the vertical slab through $R_i$ to obtain Attribute 2.

Through this procedure, we construct the dependency graph $G(\RRt)$ in $O(n \log n)$ time by iterating over all nodes in this graph.
In linear time, we can identify the connected components of $G(\RRt)$ and the regions which are both a source and sink in $G(\RRt)$. 
From Lemma~\ref{lemma:independence} we know that we can solve each connected component of $G(\RRt)$ independently and that the solutions must be concatenated through the regions that are both a source and sink. We store the connected components of $G(\RRt)$ as a doubly linked list and remove all sources and sinks from $\RRt$ to create a \emph{culled} set. 
To transform a culled set into a canonical set, we identify all sinks in the graph in linear time (by checking if $|V_i| = |H_i| = 1$) and we iterate over all regions in order of their index. Neighboring sinks get recursively grouped into a compound region and this procedure creates a canonical set in linear time. 
For each region compounding $k$ regions, we construct Attribute 5 in $O(k)$ time.
After having compound all regions, we do a linear-time scan to re-index all the (compound) regions so that all indices are consecutive and we obtain a \emph{canonical set} $\RRc$.
During this linear time scan, we identify for each $R_i$ the region of its \emph{back pointer} and \emph{forward pointer}  (Attribute 3) in logarithmic time, through searching through the vertical and horizontal decomposition. 
Moreover, whenever we compound a set $[R_i, R_{i+k}]$ into a region $R$, we make sure to remove $[R_i, R_{i+k}]$ from $G(\RRt)$ and replace it with $R$ (where all arrows pointing to a region in $[R_i, R_{i+k}]$ now point to $R$). In this way, we simultaneously create $G(\RRc)$.

Lastly, we want to obtain from a canonical set $\RRc$ its \subproblem tree $T_{\RRc}$ in $O(n)$ time using prior constructed $G(\RRc)$.
This can be done as follows: first we identify the \subproblems of $G(\RRc)$ in linear time. Then for each \subproblem $[R_i, R_j]$ of $G(\RR)$ we (temporarily) remove all outgoing arrows from $R_i$ and $R_j$ from the graph and for each node that has an arrow from $R_i$ or $R_j$ we check if it becomes a source node in constant time. This gives us the child nodes of the node that stores $[i,j]$ in the $T_{\RRc}$.
During this process, we store for each region $R_i$ a pointer to the largest interval $[i, \cdot]$ in the $T_{\RRc}$ (which must always exist) in constant additional time per region (Attribute 4).
Applying this procedure recursively takes time linear in the number of edges in $G(\RR)$, which itself is linear in the number of cells of the vertical and horizontal decomposition of $\RRt$, which concludes the lemma.
\end{proof}

\noindent
Lemma~\ref{lem:truncation} and~\ref{lem:canonicalset} and the observation that $n \le m$ immediately imply Theorem~\ref{thm:preprocessing}.

\begin{theorem}
\label{thm:preprocessing}
For any set $\RRo$ of $m$ axis-aligned, pairwise disjoint axis-aligned rectangles we can construct its trucated set $\RR$ and its canonical set $\RRc$ and $\Xi$ in $O(m \log m)$ time.
\end{theorem}

\subsection{Reconstruction phase}
\label{sub:reconstruction}
We want to run Algorithm~\ref{algo:upperbound} whilst maintaining Invariant~\ref{inv:pointer}, in $O(A(\RR, \RRc, P))$ time  (Theorem~\ref{thm:runtime}). First, we argue that the reporting (appending) step of the algorithm is correct:
\begin{restatable}{lem}{identification}
\label{lem:identification}
For any iteration $t$, for any \subproblem $[R_i, R_j]$ of $G(\RR^t)$, the point $p_i$ appears on the \pareto of $P$ if and only if $p_i$ is not dominated by $p_i^\mathit{xMax}$ or $p_j^\mathit{yMax}$.
\end{restatable}

\begin{proof}
Let $p_i$ be not dominated by $p_i^\mathit{xMax}$ and $p_j^\mathit{yMax}$, but dominated by some point $p_k$. 
Then $k < i$ or $k > j$, because $R_i$ and $R_j$ are both sources in $G(\RR^t)$.
If $k < i$ then the $x$-coordinate of $p_k$ is greater than of $p_i$, and thus $p_i^\mathit{xMax} \neq p_i$.
Then, the point $p_i^\mathit{xMax}$ has greater $x$-coordinate than $p_i$, it lies in some region $R'\not=R_i$, and since $R'$ precedes $R_i$ and contains $p_i^\mathit{xMax}$, its bottom facet must lie above the top facet of $R_i$. Thus $p_i^\mathit{xMax}$ dominates $p_i$ which is a contradiction. 
 If $j < k$ then $p_j^\mathit{yMax} \neq p_j$ and the symmetrical argument applies.
\end{proof}
\noindent
The previous lemma implies that if Invariant~\ref{inv:pointer} is maintained, we can iteratively identify points that appear on the \pareto. 
Lemma~\ref{lemma:independence} guarantees that for each iteration $t$, for each \subproblem $[R_i, R_j]$, the \pareto of $\{p_i^\mathit{xMax} \} \cup [p_i, p_j] \cup \{p_j^\mathit{yMax} \}$ is a connected subchain of the \pareto of $P$.
Hence we can safely append $p_i$ after $p_i^\mathit{xMax}$.
What remains to show is that we can maintain Invariant~\ref{inv:pointer} and identify the \subproblems of $\RR^t$ efficiently.

\subparagraph{Identifying \subproblems.}
Consider an iteration $t$ in which we handle \subproblem $[R_i, R_j]$, and let $[R_k, R_l]$ be any \subproblem of  $G(\RR^{t+1})$ that is not already a \subproblem of $G(\RR^t)$.
It must be that $i \le k \le l \le j$ (Lemma~\ref{lemma:independence}).
We need to quickly identify these new \subproblems.

\begin{lemma}
For any truncated set $\RR^t$, for any \subproblem $[R_i, R_j]$ of $G(\RR^t)$,
either $f_i(P) \in V_i$ or $f_i(P) = \Vnext{i}$.
\end{lemma}

\begin{proof}
Any region in $[R_i, R_j]$ that is dominated by a point preceding $p_i$ is dominated by $p_i^\mathit{xMax}$. The point $p_{i-1}^\mathit{xMax}$ cannot dominate $R_i$, as else $R_i$ would have been removed during a truncation.
Hence, $f_i(P)$ is $\Vnext{i}$ or a region preceding it. 
Suppose for the sake of contradiction that $f_i(P)$ is a region preceding $\Vnext{i}$ and not in $V_i$. Consider any vertical ray from a point in $R_i$, right of $p_i^\mathit{xMax}$ that intersects $f_i(P)$ (such a ray must always exist, since $f_i(P)$ precedes $\Vnext{i}$ and is not dominated by $p_i^\mathit{xMax}$).
Since $f_i(P) \not \in V_i$, this ray must also intersect a region $R' \in V_i$ (else this ray would be a line of sight to $f_i(P)$, which would imply $f_i(P) \in V_i$). However, then $R'$ must precede $f_i(P)$ which contradicts the assumption that $f_i(P)$ was the lowest-indexed region succeeding $R_i$, not dominated by $p_i^\mathit{xMax}$.
\end{proof}

\begin{restatable}{cor}{pointlocation}
\label{cor:pointlocation}
Let $\RR^t$ be a truncated set, $[R_i, R_j]$ be a \subproblem.
Given Invariant~\ref{inv:pointer} and $\Xi$, we can identify $f_i(P)$ in $O(\log |V_i(P)|)$ time using the folklore galloping search.
\end{restatable}

\begin{proof}
The datastructure $\Xi$ stores for $R_i$ the set $V_i$ as a balanced binary search tree (Attribute 1). The set $V_i(P)$ is a prefix of $V_i$ which ends at $f_i(P) \in V_i$ (or, in the case that $V_i(P) = V_i$, $f_i(P) = V_i^\mathit{next})$).
Thus, given Invariant~\ref{inv:pointer}, we can use $p_i^\mathit{xMax}$ to identify $V_i(P)$ in $O(\log |V_i(P)| )$ time by using the folklore galloping (exponential) search by Bentley and  Chi-Chih Yao.
If $V_i(P) = V_i$, we refer to $V_i^\mathit{next}$ which is stored in $\Xi$ (Attribute 2).
\end{proof}

\noindent
Next, we prove a lemma that helps us to identify the \subproblems of $G(\RR^{t+1})$:

\begin{lemma}
\label{lemma:subproblems}
Let $[R_i, R_j]$ be a \subproblem of $G(\RR^t)$ and denote by $v$ the lowest node in $T_\RR$ such that the interval $[i, j]$ is stored in $v$.
For any descendent $[a, b]$ of $v$, there is no region $R' \in [R_a, R_b]$ that is a source node in $G(\RR^{t+1})$ other than possibly $R_a, R_b, f_i(P)$ or $g_j(P)$.
\end{lemma}

\begin{figure}[t]
    \centering
    \includegraphics{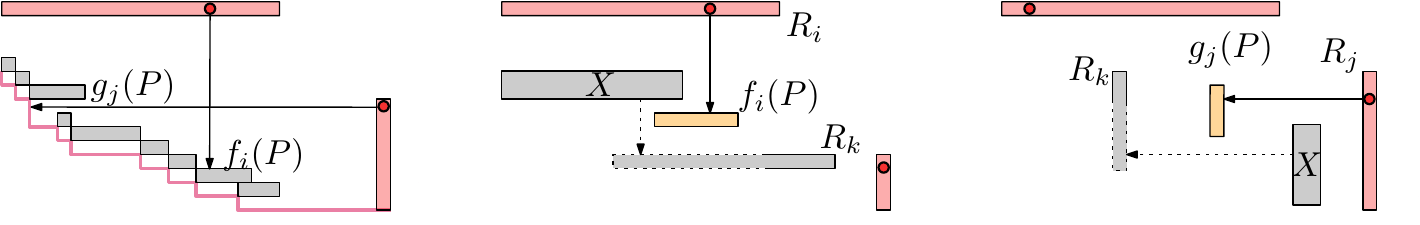}
    \caption{An illustration of the argument of Lemma~\ref{lemma:subproblems}. If $R_k$ loses the incoming arrow from $X$, there must be a directed path from $f_i(P)$ or $g_i(P)$ to $R_k$, or either $R_k = f_i(P)$, $R_k = g_j(P)$. }
    \label{fig:nextsubproblem}
    \vspace{0cm}
\end{figure}

\begin{proof}
If $f_i(P)$ equals or succeeds $g_j(P)$ then per definition of $f_i(P)$ and $g_j(P)$ all regions in $(R_i, R_j)$ apart from $f_i(P) = g_j(P)$ are dominated and therefore removed after truncation of $\RR^{t+1}$. Hence, they cannot be sources in $G(\RR^{t+1})$ (Figure~\ref{fig:nextsubproblem}).
Let $[a,b]$ be a descendent of $v$, $R_k$ be a region with $k \in \langle a, b \rangle$ succeeding $f_i(P)$ and preceding $g_j(P)$. 
Per construction of $T_\RR$ each such $R_k$ has at least one incoming arrow from a region $X \in [R_a, R_b]$.
The region $R_k$ can only become a source in $G(\RR^{t+1})$ if either $p_i$ or $p_j$ dominates $X$ (else, $X$ was dominated by $p_i^\mathit{xMax}$ or $p_j^\mathit{yMax}$ before iteration $t$ and does not exist in $G(\RR^{t})$).

We consider the case where $p_i$ dominates $X$ (Figure~\ref{fig:nextsubproblem}).
If $p_i$ dominates $X$, then $X$ lies strictly left of the vertical line through $p_i$, and $R_k$  intersects the vertical halfslab of $X$.
Similarly if $f_i(P) \neq R_k$ then $R_k$ must lie at least partly right of the vertical line through $p_i$ and below the bottom facet of $f_i(P)$. This means that if $R_k$ lies in the vertical halfslab of $X$ then it must also lie in the vertical halfslab of $f_i(P)$. The region $f_i(P)$ is therefore a node in $G(\RR^t)$ with a directed path to $R_k$, so $R_k$ is not a source node in $G(\RR^{t+1})$.
\end{proof}

\subparagraph{Algorithm~\ref{algo:upperbound} runtime.}
We further specify the iterative procedure of our algorithm. Our algorithm maintains a queue of \subproblems. In iteration $t$, 
we dequeue a \subproblem $[R_i, R_j]$ of $G(\RR^t)$ and we denote by $v$ the lowest node in $T_\RR$ such that the interval $[i, j]$ is stored in $v$.
We can obtain $v$ in constant time via Attribute~\ref{attr4}.
By Lemma~\ref{lemma:independence}, processing $[R_i, R_j]$  does not affect other \subproblems which are in the queue before we process $[R_i, R_j]$.
If the algorithm has not yet retrieved $p_i$ nor $p_{i-1}^\mathit{xMax}$, it retrieves both points using Invariant~\ref{inv:pointer} in $2C$ time and computes $p_i^\mathit{xMax}$ in constant time.  Similarly we compute $p_j^\mathit{yMax}$ in with at most $2C$ additional time.
By Lemma~\ref{lem:identification}, we check in $O(1)$ time if $p_i$ and $p_j$ appear on the \pareto, and if so we add them as the respective successor of $p_{i-1}^\mathit{xMax}$ or predecessor $p_{j+1}^\mathit{yMax}$. 
If we have just retrieved $p_i$, we use galloping search to identify $f_i(P)$ in $O(\log |V_i(P)|)$ time (Corollary~\ref{cor:pointlocation}), we set the back pointer (Attribute~\ref{attr3}) to \emph{null} and (for later use) we store a reference in $R_i$ to $f_i(P)$. 
If we did not retrieve $p_i$ this iteration, we retrieved it in a prior iteration and we use the pre-stored result $f_i(P)$ in $O(1)$ time. 
We do the same for $g_j(P)$ in $O(C + \log |H_j(P)|)$ time. We briefly remark the following claim.

\begin{restatable}{lem}{sourceidentification}
\label{lemma:sourceidentification}
Let $[R_i, R_j]$ be a \subproblem of $G(\RR^t)$ and $f_i(P)$ precede $g_j(P)$. Then the region $f_i(P)$ is a source in $G(\RR^{t+1})$ if and only if: (1) the forward pointer of $f_i(P)$ is \emph{null} or (2) the region resulting from the forward pointer has been retrieved in an iteration $t' < t$.
\end{restatable}

\begin{figure}[tb]
    \centering
    \includegraphics{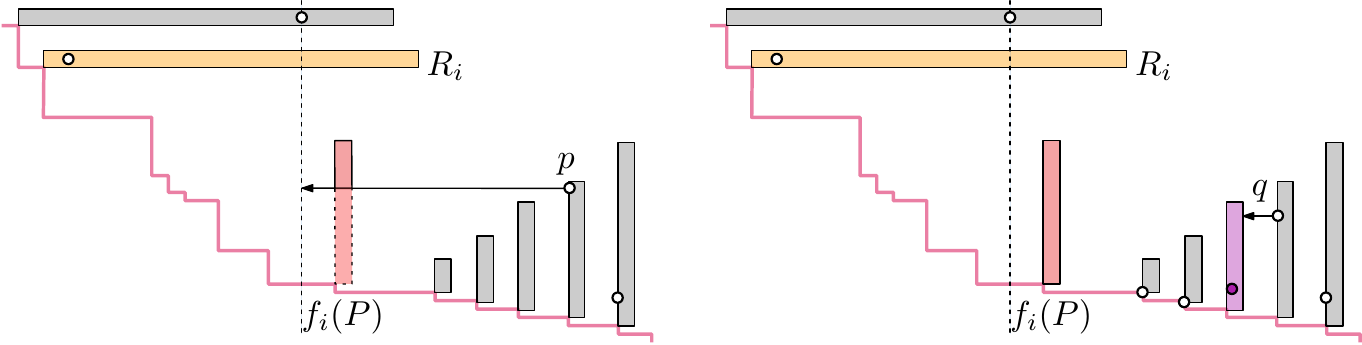}
    \caption{ Left: the first case of the proof of Lemma~\ref{lemma:sourceidentification}, where $p$ must dominate the remaining regions with an arrow to $f_i(P)$. Right: the second case, where either $q$ sees $f_i(P)$, dominates $f_i(P)$ or the purple region keeps its horizontal arrow to $f_i(P)$. }
    \label{fig:sourceproof}
    \vspace{0cm}
\end{figure}

\begin{proof}
Suppose that the pointer is \emph{null} and suppose that there is no region $R_k$ for which $f_i(P) \in H_k(P)$. 
Then $f_i(P)$ has no incoming horizontal arrows. 
If there is a region $R_k$ for which $f_i(P) \in H_k (P)$ then there is a point $p$ retrieved in an iteration earlier such that $p$ is horizontally visible from $f_i(P)$ that set the pointer to \emph{null} Figure~\ref{fig:sourceproof}, Left. The point $p$ dominates all remaining regions with a horizontal arrow  to $f_i(P)$.
If the region resulting from the forward pointer has been retrieved in an iteration $t'< t$, all regions with a horizontal pointer to $f_i(P)$ must have been considered by the algorithm, so $f_i(P)$ is not dominated.
By definition, all regions preceding $f_i(P)$ in $[R_i, R_j]$ are dominated by $p_i^\mathit{xMax}$, thus, if  $f_i(P)$ has no incoming horizontal arrows it must be a source in $G(\RR^{t+1})$.

If the pointer is not \emph{null} and the region resulting from the forward pointer has not yet been retrieved in an earlier iteration then $f_i(P)$ must have at least one incoming horizontal arrow. 
Indeed, suppose that all regions with a horizontal pointer to $f_i(P)$ that are not yet retrieved are dominated by a point $q$ retrieved prior to the current iteration. Then either $q$ dominates $f_i(P)$, contradicting the assumption that $f_i(P)$ precedes $g_j(P)$, or the retrieval of $q$ would have set the forward pointer of $f_i(P)$ to \emph{null}.  
\end{proof}

\noindent
For ease of exposition, we assume $f_i(P)$ and $g_j(P)$ are not compound regions. For compound regions, we refer to Appendix~\ref{appx:reconstruction}. We distinguish between two cases based on which children of $v$ contain $f_i(P)$ and $g_j(P)$ (Figure~\ref{fig:algorithmicrecursion}).
Note that we never add a $\subproblem$ $[R_a, R_b]$ if $a = b+1$ (as such a \subproblem does not satisfy the premise of Theorem~\ref{thm:runtime}).
Instead, we charge retrieving and comparing $p_a$ and $p_b$ immediately with at most $4C$ overhead. 

 \subparagraph{Case 1: $f_i(P)$ and $g_j(P)$ are contained in the same grandchild $[a,b]$ of $v$.}
 We check in constant time whether $f_i(P)$ and $g_j(P)$ are sources in $G(\RR^t)$ (by Lemma~\ref{lemma:sourceidentification}).
Note that either $f_i(P)$ or $g_j(P)$ must be a source.
 Let $R_k = f_i(P)$ and $R_l = g_j(P)$.
 \begin {itemize}
 \item If both $f_i(P)$ and $g_j(P)$ are sources, then by Lemma~\ref{lemma:subproblems} the only three \subproblems in $G(\RR^{t+1})$ and $[R_i, R_j]$ are:
 $[R_i = p_i, R_k]$, $[R_k,R_l]$ and $[R_l, R_j = p_j]$.
 In this case $p_{k-1}^\mathit{xMax} = p_i^\mathit{xMax}$ and $p_{l+1}^\mathit{yMax} = p_j^\mathit{yMax}$. If $k = l-1$, we immediately retrieve $p_k$ and $p_l$ in 2C time as the aforementioned overhead. Else we add to $[R_k, R_l]$ a reference to $p_{k-1}^\mathit{xMax}$ and $p_{l+1}^\mathit{yMax}$ to maintain Invariant~\ref{inv:pointer} and add the \subproblem $[R_k, R_l]$ to the queue.
 \item  If $f_i(P)$ is a source and $g_j(P)$ is not, by the same reasoning the only \subproblems are $[R_i, R_k]$ and $[R_k, R_j]$.
 We check if $k = j-1$ as before. If not, we maintain Invariant~\ref{inv:pointer} in constant time just as above by adding $[R_k, R_k]$ to the queue with a reference to $p_i^\mathit{xMax}$. 
 \item This case is symmetric to the previous, as $f_i(P)$ is not a source and $g_j(P)$ is.
 \end {itemize}
 
 \subparagraph{Case 2: $f_i(P) \in [R_a, R_b]$ and $g_j(P) \in [R_e, R_f]$ for distinct children $[a,b]$ and $[e, f]$ of $v$.}
 In this case, per construction of $T_{\RRc}$, each child $[c, d]$ of $v$ with $b \le c < d \le e$ is a \subproblem of $G(\RR^{t+1})$. 
 We wish to briefly note, that either $c < d-1$, or $[c, d]$ neighbors a child of $v$ for which this is true (else, regions could have been compounded). Hence by Theorem~\ref{thm:runtime} if $c = d-1$ we charge $2C$ time to the neighbor to immediately retrieve $p_c$ and $p_d$ and possibly add them to $\Xi^*$ (again as the aforementioned overhead). 
 If $c < d-1$, then per construction of $T_{\RRc}$, the point $p_c$ appears on the \pareto of $P$. 
 Note that since $[c, d]$ is a child of $v$, $p_{c-1}^\mathit{xMax}$ can only be $p_{c-1}$ or $p_i^\mathit{xMax}$.
 We charge $O(1)$ time to the future processing of $[R_c, R_d]$ to provide four pointers to $[R_c, R_d]$ (to maintain Invariant~\ref{inv:pointer}) and add $[R_c, R_d]$ to the queue.

 \begin{figure}[tb]
    \centering
    \includegraphics{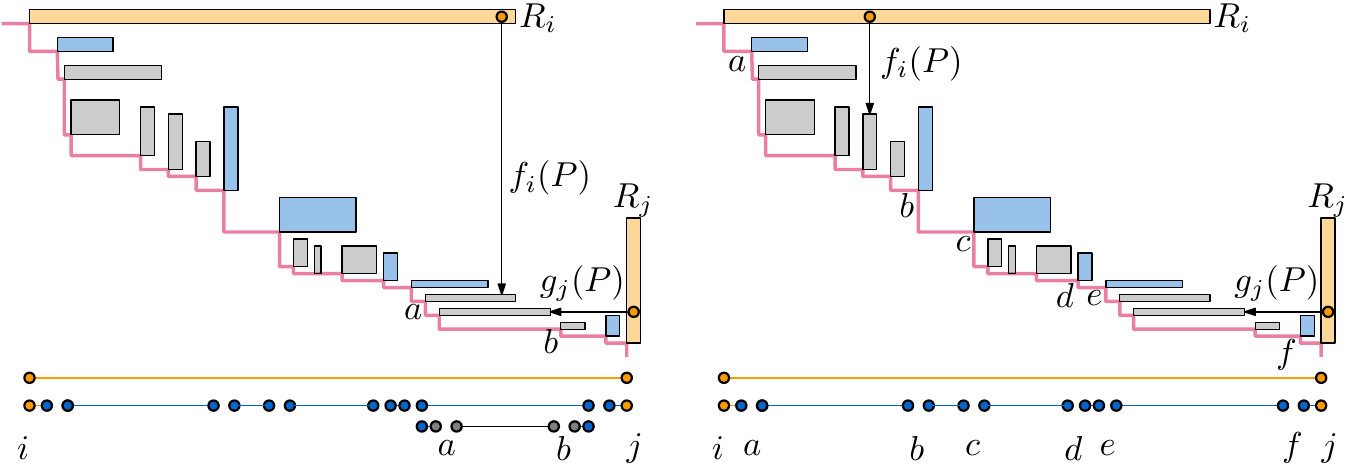}
    \caption{Left: Case 1 where $f_i(P), g_j(P)$ lie in the same grandchild $[a,b]$. Right they don't. }
    \label{fig:algorithmicrecursion}
    \vspace{0cm}
\end{figure}

 What remains is to handle $[a,b]$ and $[e,f]$ and we describe the procedure for $[a,b]$.
  We check in constant time if $f_i(P)$ is a source using Lemma~\ref{lemma:sourceidentification}. If it is, then by Lemma~\ref{lemma:subproblems} the only \subproblems of $G(\RR^{t+1})$ contained in $[R_i, R_b]$ are $[R_i, f_i(P)]$ and $[f_i(P), R_b]$.
  We briefly check if $[f_i(P), R_b]$ is a \subproblem of length 2. If so we retrieve the corresponding points to see if they appear on the \pareto. Else we add $[f_i(P), R_b]$ to the queue in constant time via the same procedure as Case 1.
  If $f_i(P)$ is not a source, then $[R_i, R_b]$ is the only \subproblem of $G(\RR^{t+1})$ in $[R_i, R_b]$ and we handle it similarly.  We conclude:
\begin{theorem}
Algorithm~\ref{algo:upperbound} constructs $\Xi^*$ in $O(A(\RR, \RRc, P)) = \Theta(\CP(\RR, P))$ time. 
\end{theorem}

\bibliography{main}

\clearpage
\appendix

\section{Reviewing lower bounds}
\label{appx:entropy}

The folklore worst-case lower bound definition of an algorithmic problem $\mathcal{P}$ with input $X$ is:
\[
\textnormal{Worst-case lower bound}(\mathcal{P}) := \min_{A} \max_{X} \textnormal{Runtime}(A, X)\,,
\]
\noindent
where each $A$ is an algorithm that solves $\mathcal{P}$ for some definition of solving.
Afshani, Barbay and Chan~\cite{afshani2017instance} observe that there are three common techniques to prove lower bounds within computational geometry:
\begin {itemize}
  \item direct arguments based on counting, or {\em information theory};
  \item topological arguments, as used by e.g. Yao~\cite{yao1981lower} or Ben-Or~\cite{ben1983lower} (sometimes referred to as \emph{algebraic decision tree} arguments); or
  \item arguments based on Ramsey theory, as used by e.g. Moran, Snir and Manber~\cite{moran1985applications}. 
 \end {itemize}
 The latter two techniques decompose algorithms into decision trees and reason about their depth.
 In traditional computation models decisions are binary; therefore,
 without additional information about the decision tree structure of the specific problem $\mathcal{P}$, the best possible lower bound on its tree depth is $\Omega(\log (\#\text{leaves}))$, which is equivalent to the information-theoretic bound. 
 We mention that an additional technique for obtaining lower bounds is an adversarial argument as by Erickson~\cite{erickson1995lower} or the more recent Chan~\cite{chan2010comparison}. 
 Here, we restrict our attention to information-theoretic arguments.

 \subparagraph{Models of computation.}
Applying these techniques to bound the running time of the algorithms $A$, requires a precise definition of the model of computation used for the algorithmic analysis.
The classical argument by Ben-Or~\cite{ben1983lower} assumes that the computation can be modeled by an algebraic decision tree, where in each node a binary decision is taken at which the algorithm branches based on an algebraic test.

Afshani, Barbay and Chan investigate a stronger definition for an algorithmic lower bound. They reason that the computational power that comes from the abstract algebraic decision tree model, where algebraic test functions are only bounded in the number of arguments and not their degree, is too large for a more fine-grained analysis of algorithmic running time.
They restrict the class of algorithms that they consider for their competitive analysis to algebraic decision trees where each test is a multilinear  function (a function that is linear, separate in each of its variables) with a constant number of variables. 
We share the sentiment that a computational model that allows arbitrary algebraic computations in constant time is unrealistically powerful, but note that the alternative model is perhaps too restrictive, as it becomes difficult, if not impossible, to express computations such as higher-dimensional range searching using only multilinear functions.

Recently, Erickson, van der Hoog and Miltzow~\cite{hoog2020FOCS} note that computations that involve data structures do not only need to make decisions, but also need to be able to access memory. Memory is inherently discrete: a model that supports only real-valued algebraic decisions can either not access memory, or has the ability to access discrete values with real-valued computations which would imply that $P = PSPACE$~\cite{schonhage1979power}.
Fueled by the desire to analyse algorithms within computational geometry, they (re)define the real RAM. We use their definition of RAM to be able to define lower bounds for the preprocessing model (as the preprocessing model inherently can access memory as it needs to be able to use an auxiliary data structure $\Xi$). 
For completeness, we summarize their definition and how it enables an information theoretic lower bound, even when dealing with a pre-stored structure $\Xi$ at the end of this section.

\subparagraph{Better than worst-case optimality.}
A natural more refined lower bound than the worst-case lower bound is the instance lower bound. Given an algorithmic problem $\mathcal{P}$ with input $X$, the \emph{instance lower bound} is defined as:
\[
\textnormal{Instance lower bound}(\mathcal{P}, X) := \min_{A} \textnormal{Runtime}(A, X)\,.
\]
We recall the example in the introduction where we perform a binary search to see whether a value $q$ is contained in a sorted sequence of numbers $X$. 
For each instance $(X, q)$, there exists a ``lucky'' algorithm that guesses the location of $q$ in $X$ in constant time. Thus, the instance lower bound for binary search is constant, even though there is no algorithm that can perform binary search in constant time in a comparison-based RAM model.
Fine-grained algorithmic analysis is desirable, yet instance optimality is unobtainable. It is therefore unsurprising that there is a rich tradition of finding algorithmic analyses that capture an algorithmic performance that is better than worst-case optimality. 
Many attempts parametrize the algorithmic problem, to better enable its analysis. For example, there is {\em output-sensitive} analysis as used by Kirkpatrick and Seidel~\cite{kirkpatrick1985output} where the algorithm runtime depends on the size $k$ of the output. Other parameters can include geometric restrictions such as {\em fatness}, the {\em spread} of the input, or the number of {\em reflex} vertices in a (simple) polygon. 
Such parameters are hard to apply in the preprocessing model with implicit representation, as the auxiliary structure $\Xi$ allows one to bypass the natural lower bound that these parameters bring. For example: an output-sensitive lower bound is not applicable, as output of any size can be computed in the preprocessing phase to be referred to in the reconstruction phase in $O(1)$ time.

\subparagraph{Better than worst-case optimality without additional parameters.}
Afshani, Barbay and Chan propose an alternative definition of instance optimality which is not inherently unobtainable.
They restrict the algorithms $A$ that solve $\mathcal{P}$ and consider the input $I$ together with a permutation $\sigma$.
They analyse the running time of $A$, conditioned on that it receives input $X$ in the order given by $\sigma$. They then compare algorithmic running time based on the worst choice of $\sigma$:
\[
\textnormal{Instance lower bound in the {\em order oblivious} setting}(\mathcal{P}, X) := \min_{A} \max_{\sigma} \textnormal{Runtime}(A, X, \sigma)\,.
\]
Intuitively, a permutation $\sigma$ can force the algorithm to make poor decisions by placing the input in a bad order and they assume that an algorithm receives ``the worst order of processing the input'' to avoid the unreasonable computational power that a guessing algorithm has. 
The instance lower bound in the order oblivious setting for our binary search example would be $\Omega(n)$, as there exists a $\sigma$ for which $X$ is not a sorted set. Given $q$ and $(X, \sigma)$, any algorithm then has to spend linear time to check if $q$ is in $X$.

This definition of lower bound would strictly speaking be applicable to the preprocessing model: given $P$ and a permutation $\sigma$ an algorithm can then only retrieve points in the order $\sigma$. 
However, we would argue that this lower bound is not very compatible with the spirit of the model. Per definition, one is free to preprocess $\RR$, Therefore, during preprocessing it would not be unreasonable for an algorithm to decide on a favourable order to retrieve the points in $P$. 
This is why, amongst many alternative stricter-than-worst-case lower bound definitions, we propose another, specifically for the preprocessing model.
\[
\textnormal{Uncertainty-region lower bound}(\mathcal{P}, \RR) := \min_{(A, \Xi)} \max_{ (P \textnormal{ respects } \RR ) } \textnormal{Runtime}(A, \Xi, \RR, P)\,,
\]
Denote for any fixed algorithmic problem $\mathcal{P}$, by $L(\RR)$ the number of combinatorially distinct outcomes of $P$ given $\RR$. In the remainder of this section we recall the RAM definition of~\cite{hoog2020FOCS} to show that regardless of $(A,\Xi)$, $\Omega(\log L(\RR))$ is an uncertainty region lower bound for the time required by $A$ to solve $\mathcal{P}$.

\subparagraph{Recalling the real RAM definition.}
If the reader is confident in the ability of the RAM model to support such a lower bound, we advise the reader skips ahead.
Erickson, van der Hoog and Miltzow define the real RAM in two steps. First, they define computations based on the (discrete) word RAM, so that discrete memory can be accessed without unreasonable computational power. Then, they augment the word RAM with separate real-valued computations that only work on values stored within the discrete memory cells.
Their operations include memory manipulation, real arithmetic and comparisons (which verifies if the real value stored in a memory cell is greater than $0$).
For an extensive overview of the computations that they allow, we refer to Table 1 in \cite{hoog2020FOCS}.
They say a program on the real RAM consists of a fixed, finite indexed sequence of read-only instructions.  The machine maintains an integer \emph{program counter}, which is initially equal to $1$.  At each time step, the machine executes the instruction indicated by the program counter.
Every real RAM operation increases the program counter by one, apart from a comparison operation which ends in a goto statement that can set the program counter to any discrete value.
This model thereby immediately allows the classical information theoretic lower bound argument, even if there is some pre-stored data $\Xi$ within memory.
Indeed, let $\mathcal{P}$ be an algorithmic problem such that there are $L$ distinct outcomes and fix a program (algorithm) that reports the correct outcome.
Each outcome may be described by the sequence of instructions that lead to it, together with a \emph{halt} instruction that tells the program to stop and output the result.
Hence, the program only terminates on the correct outcome, if it arrived there via a \emph{goto} statement from a comparison instruction (all other instructions only increase the program counter by 1, hence without comparisons the algorithm terminates at the first outcome in the sequence). 
It follows, that any sequence of instructions can be converted into a binary tree where each node is a comparison instruction and where the leaves of the tree are lines in the sequence that store an outcome with a halt instruction.
Hence regardless of $\Xi$, there is an outcome stored as a leaf in the tree where the program that requires $\Omega(\log L)$ comparison instructions until it arrives at that leaf.

\section{Handling compound regions}
\label{appx:reconstruction}
We describe the algorithmic procedure for when Algorithm~\ref{algo:upperbound} encounters a \subproblem $[R_i, R_j]$ where $f_i(P)$ or $g_i(P)$ is a compound region.
Let $f_i(P)$ be a compound region. Then per definition $f_i(P)$ is a sink in the original graph: $G(\RR^0)$. Consequently, the region $R'$ in the canonical set $\RR^0$ that succeeds $R$ must have no more remaining incoming vertical arrows (as else, $R$ would not have been visible from the just processed $R_i$).
The region $R'$ itself cannot be a compound region, since else $R$ and $R'$ could have been compounded together. We set $f_i(P)$ to be $R'$ instead, and continue as normal.

We set the compound region $R$ aside, with a reference to $p_i^\mathit{xMax}$ and add it to a separate queue that we handle at the algorithm's termination in $O(1)$ time. We charge this $O(1)$ time to this iteration $t$ where we added it to the special queue. Per definition, for each region $R_i$, there is a unique $f_i(P)$, so $R_i$ gets charged at most once in this manner.
It is possible that in a later iteration $t'$, when a \subproblem $[R_{i'}, R_{j'}]$ is considered by Algorithm~\ref{algo:upperbound}, the region $R$ is $g_{j'}(P)$. In this case, we do not add $R$ to the queue again but we do store a reference to $p_{j'}^\mathit{yMax}$ and we charge $[R_{i'}, R_{j'}]$, $O(1)$ time for storing this reference.

For any compound region $R$, that is not dominated by a point in $P$, there must be an iteration $t$ where a \subproblem is considered such that $f_i(P) = R$ or $g_j(P) = R$ and thus it must be in the special queue.  When we process the special queue, we do the following: 
  we use $p_i^\mathit{xMax}$ to identify the prefix of the original regions stored in $R$ that are dominated by points preceding $R$ in $O(\log |V_i(P)|)$ time using galloping search (we charge the prior $f_i(P)$, and just as above a region can only get charged once).
  
  At this point, we wish to briefly remark upon any possible ambiguity regarding the runtime $O(\log |V_i(P)|)$. 
  In the premise of Theorem~\ref{thm:runtime} we defined the sets $V_i(P)$ as subsets of the truncated set $\RR$, not the canonical set $\RR^0 = \RRc$ that serves as the input of the algorithm.
  Note that $O(\log |V_i(P)|)$ is smaller than $O(\log |V_i^*(P)|)$ where $V_i^*(P)$ is a subset of $\RRc$ since $\RRc$ can compound regions in $V_i(P)$ together. Throughout Section~\ref{sub:reconstruction}, 
  we performed a galloping search over the outgoing edges in the graph $G(\RRc)$, hence we spent $O(\log |V_i^*(P)|) \le O(\log |V_i(P)|)$ time per search.
  Here, we perform a galloping search over regions in $V_i$ that are compounded (not in $\RRc$), and this is the first point where we use the larger $O(\log |V_i(P)|)$ runtime.
  We wish to emphasise that the runtime of Section~\ref{sub:retrievals} is hereby correct: as $O(\log |V_i(P)|)$ is an over-estimation of the actual time spent on the galloping search. We continue the argument:

  Whenever $g_j(P) = R$, we similarly use $p_j^\mathit{yMax}$ to identify the suffix of the original regions stored in $R$ that are dominated by points in $P$ succeeding $R$.
  For the at most $2$ regions that are intersected by the vertical line through $p_i^\mathit{xMax}$ and the horizontal line through $p_j^\mathit{yMax}$ respectively, we explicitly retrieve their points in order to determine whether they are dominated or not.
  We charge this $2C$ retrieval time  to $R_i$ and $R_j$. 
  By Lemma~\ref {lem:contiguous}, the remaining sequence of original regions (if any) must appear on the \pareto, and we do not need to retrieve their points.
  When the algorithm terminates, we append the non-dominated interval in constant time by providing the pointers in the array of Attribute 5, and we charge this constant time to the aforementioned iteration.

\end{document}